\documentclass[aps,prxquantum,twocolumn,groupedaddress]{revtex4-2}

\usepackage{amsmath}
\usepackage{amsthm}
\usepackage{amsfonts}
\usepackage{mathtools}
\usepackage{amsmath,amssymb}
\usepackage{graphicx}
\usepackage{latexsym, xcolor}

\newtheorem{theorem}{Theorem}
\newtheorem{lemma}[theorem]{Lemma}

\newtheorem{proposition}[theorem]{Proposition}

\begin{document}


\title{Quantum Optimal Control of a Lambda System in the Density Matrix Formulation}


\author{Julia Cen}
\email[]{julia.cen@outlook.com}
\affiliation{Iowa State University, Ames, IA 50011}
\author{Domenico D'Alessandro}
\email[]{daless@iastate.edu}
\affiliation{Iowa State University, Ames, IA 50011}


\date{\today}

\begin{abstract}
In various physical implementations of quantum information processing, qubits are realized in a {\it Lambda} type system configuration as two stable lower energy levels coupled indirectly via an unstable higher energy level, that is, in comparison, a lot more susceptible to decoherence. We consider the quantum control problem of optimal state transfer between two isospectral density matrices, over an arbitrary finite time horizon, for the quantum {\it Lambda} system. The cost considered is a compromise between the energy of the control field and the average {\it occupancy} in the highest energy level. We apply a geometric approach that combines the use of the {\it Pontryagin Maximum Principle},   a {\it symmetry reduction technique} to reduce the number of parameters in the resulting optimization problem, and several auxiliary techniques to bound the parameter space in the search for the optimal solution.  We prove several properties of the optimal control and trajectories for this problem including their {\it normality} and {\it smoothness}. We obtain  a system of differential equations that must be satisfied by the  optimal pair of control and trajectory we treat in detail, with numerical simulations,  and solve  a case study involving a Hadamard-like transformation. Our techniques can be adapted to other contexts and promise to push to a more consequential level, the application of geometric control in quantum systems.
\end{abstract}


\maketitle

\section{Introduction}

Due to the detrimental effect of {\it decoherence} from environmental interactions, a key challenge in quantum information processing is the ability to perform fast, high-fidelity quantum operations on a quantum bit, or qubit. This necessitates the study of control for open quantum systems. One important operation is to drive a state to a desired final configuration coherently. An approach to do this is to consider the control problem for the quantum master equation (QME) (see, e.g, \cite{Petruccione}).  Besides the issue of validity of the QME in the control theory settings, where certain conditions for the system-bath interaction have to be satisfied \cite{conEdmo1}, control theory for open quantum systems is not as well developed, and, in general, more complicated than the corresponding theory for closed systems \cite{pmp1,pmp2,pmp3,pmp4,pmp5,pmp6,pmp7,closed1,pmp8,pmp9,pmp10}. Even some fundamental aspects such as controllability \cite{Dirrrev} are not so firmly established since some structure, and in particular unitarity of the evolution, is lost. An alternative approach to control with decoherence minimization is to remain in the closed system setting and set up the control problem so as to minimize the occupancy of the state in regions of the state space where the effect of the environment is stronger. From this perspective, optimal control theory offers a natural formulation, as the occupancy in these regions can be modeled as a cost in an optimal control problem and treated as such. This is the approach we follow in this paper.

Quantum {\it Lambda} systems are an important, prototypical class of models appearing in quantum information processing platforms ranging from NV centers \cite{nv1,nv2, nv3}, trapped ions \cite{ti1,ti2, ti3}, quantum dots \cite{qd1,qd2},
Rydberg atoms \cite{ra1}, superconducting qubits \cite{sc1,sc2}, rare-earth ions \cite{re1}, NMR \cite{nmr1}, and more \cite{Vezvaee}. They have interesting properties and applications including electromagnetically induced transparency \cite{Atac, Ma}, quantum memories \cite{Ma, qm1, qm2, qm3}, quantum cloning \cite{qc1, qc2} and quantum metrology \cite{qs1,qs2}. This system is a three level quantum system where the lowest two levels have to be controlled with a control that does not couple them directly but couples them through a third, higher energy level. This higher energy level is the most subject to decoherence and therefore occupancy in such a level has to be minimized. This has motivated protocols such as {\it STIRAP (Stimulated Raman Adiabatic Passage)} (see, e.g., \cite{Bergmann}) to obtain, over a long time horizon, a population transfer between the two relevant energy levels with small occupancy in the highest energy level. While such adiabatic techniques can achieve nearly perfect transfer of population, long evolution times are a problem. In this study, we propose a non-adiabatic, arbitrarily fast protocol whilst minimizing energy of controls and occupancy in the highest energy level. This provides a pathway to design high-fidelity quantum operations within limited qubit coherence times. In a recent contribution \cite{conYasemin}, we have solved the problem of minimum occupancy population transfer, in {\it finite time}, within the framework of optimal control theory, for {\it pure states}. In this paper, we extend this result to the case of {\it mixed states}, by using tools of geometric control theory. In particular, our initial and final states are  general isospectral density matrix for the lowest two energy levels, that is, matrices of the form $\rho:=\begin{pmatrix} 0 & 0 \cr 0 & \rho_1 \end{pmatrix}$ with $\rho_1$ a $ 2 \times 2$ density matrix. The treatment is more complicated in this  density matrix case and involves some some extra ingredients to cope with the fact that the dimension of the problem is significantly higher.

Our approach is to apply techniques of {\it geometric control theory} (see, e.g., \cite{Agrachev}, \cite{Jurdjevic}) and in particular the {\it Pontryaging Maximum Principle} (PMP)  (see, e.g., \cite{BoscaRev}, \cite{Mybook},  for an introduction in the context of quantum systems) to obtain a system of differential equations which have to be satisfied by any pair of optimal control and trajectory. This approach reduces the problem of the search for optimal control on an (infinite dimensional) space of functions to a finite dimensional space of parameters, the initial conditions in this set of differential equations and certain parameters which appear in the equations. We shall apply  symmetry reduction techniques to further reduce the space  of the optimal parameters. We shall also obtain bounds on such parameters in order to reduce the search to a bounded  set. After this is done the search has to be done numerically by integrating the system of differential equation for various values of the parameters, selecting the parameters that lead to the desired final condition, and then, among these, selecting the parameters and therefore the optimal control which gives the minimum cost.  

The paper is organized as follows. In Section \ref{SoP} we mathematically state the problem we want to solve and set up some of the notations we use in the rest of the paper. We also highlight some assumptions we can make without loss of generality.  Dealing with density matrices one can write the corresponding differential equations in matrix form or by {\it vectorization} of the matrix in an appropriate (orthonormal) basis. This second approach lends itself more naturally to the application of the conditions of the PMP. We shall switch between these two approaches in the course of the paper.  In  Section \ref{PMPsec}  we recall the relevant background in geometric control theory and the statement of the Pontryagin Maximum  Principle which we shall use to solve our optimal control problem. In Section \ref{SysT} we prove that, for our problem, all the extremals are normal (cf. definitions in Section \ref{PMPsec}). This has several consequences, including the fact that the optimal control and trajectory are analytic functions. In this section we derive a system of differential equations that have to be satisfied by the optimal control and trajectory. This is the fundamental system which is used in the search for the optimal control. In Section \ref{Red} we prove several properties of the problem at hand with the goal of reducing the number of parameters in the search for the optimal control. As an illustration of the results obtained in Section \ref{CS} we solve with the help of numerical simulations a case study which involves the implementation of an Hadamard-like gate. Auxiliary results and more complicated proofs are reported in the Appendix.


\section{The model and statement of the problem}\label{SoP}

 The class of optimal control problems we are interested in is for  quantum  systems governed by  the Liouville-von Neumann equation
 \begin{equation}\label{LK}
\dot \rho=[H_0, \rho] +\sum_{j=1}^m [H_j, \rho] u_j, 
\end{equation} 
where $H_0,H_1,...,H_m$ are skew-Hermitian matrices, $\{u_j\}$ the control functions assumed measurable and bounded  over  a { finite interval} $[0,T]$ which is a priori specified.  The   task is to find control functions which  drive the state $\rho$ in (\ref{LK}) from an initial (known) state $\rho_{in}$ to a desired state $\rho_{des}$ and at the same time  minimize a cost functional which depends on the particular situation at hand. In this paper we will solve the problem of optimal control for a three-level quantum system in the {\it Lambda} configuration. Here $H_0$ in (\ref{LK}) is a diagonal matrix $H_0=\texttt{diag}(iE_1,iE_2,iE_3)$ with $E_1> E_2> E_3$ and $m=4$ with  the matrices $H_1,...,H_4$ having only the entries $(1,2)$ (and $(2,1)$) or $(1,3)$ (and $(3,1)$) different from zero so as to `connect'  the level $2$ and $3$ with the level $1$ but not with each-other. The states $|2\rangle$ and $|3\rangle$ are the `logical' states, that is, the states of interest but any state transformation has to necessarily involve the highest energy level state $|1\rangle$. This  is however an undesired state because it is the most subject to de-coherence. The goal then is to achieve a state transfer on the lowest two levels minimizing the average population in the interval $[0,T]$ in the highest energy level. Since we have allowed full control between the state $|1\rangle$ and the states $|2\rangle$ and $|3\rangle$, one can show that such an average population which we call here `{\it occupancy}'  can be made arbitrarily small by using a very large control field (cf. \cite{conYasemin}). Therefore a meaningful cost to be minimized is a compromise between the average magnitude of the control field (the energy) and the occupancy, which is what we shall consider here.

To formalize our problem further, we need to introduce some notations which we will use in the rest of the paper. We define the matrices 

\begin{widetext}
    \begin{equation}\label{base}
G_1:=\begin{pmatrix} 1 & 0 & 0 \cr 0 & 0 & 0 \cr 0 & 0 & 0 \end{pmatrix}, \qquad G_2:=\begin{pmatrix} 0 & 0 & 0 \cr 0 & 1 & 0 \cr 0 & 0 & 0 \end{pmatrix}, \qquad 
G_3:=\begin{pmatrix} 0 & 0 & 0 \cr 0 & 0 & 0 \cr 0 & 0 & 1 \end{pmatrix},
\end{equation}
$$
G_4:=\frac{1}{\sqrt{2}} \begin{pmatrix} 0 & 1 & 0 \cr 1 & 0 & 0 \cr 0 & 0 & 0 \end{pmatrix}, \quad G_5:=\frac{1}{\sqrt{2}} \begin{pmatrix} 0 & i & 0 \cr -i & 0 & 0 \cr 0 & 0 & 0 \end{pmatrix}, \quad 
G_6:=\frac{1}{\sqrt{2}} \begin{pmatrix} 0 & 0 & 1 \cr 0 & 0 & 0 \cr 1 & 0 & 0 \end{pmatrix}, $$

$$G_7:=\frac{1}{\sqrt{2}} \begin{pmatrix} 0 & 0 & i \cr 0 & 0 & 0 \cr -i & 0 & 0 \end{pmatrix}, \qquad 
G_8:=\frac{1}{\sqrt{2}} \begin{pmatrix} 0 & 0 & 0 \cr 0 & 0 & 1 \cr 0 & 1 & 0 \end{pmatrix}, \qquad G_9:=\frac{1}{\sqrt{2}} \begin{pmatrix} 0 & 0 & 0 \cr 0 & 0 & i \cr 0 & -i  & 0 \end{pmatrix}. 
$$
\end{widetext}

to be an orthonormal basis\footnote{$u(3)$ ($su(3)$)  is the Lie algebra of $3  \times 3$ skew-Hermitian matrices (with zero trace). $iu(3)$ is the real vector space of $3 \times 3$ Hermitian matrices.}  in $iu(3)$ with the Frobenius inner product given by $\langle A, B\rangle:=Tr(AB^\dagger)$. We also provide a  commutation table for these matrices  as it will be useful in the following computations.

\vspace{2cm}

\begin{widetext}

\begin{table}[htbp]
\begin{center}
 \begin{tabular}{ | c | c | c | c | c | c | c | c | c |c|}
 \hline
$[ \cdot , \cdot]$  & $ iG_1$ & $ iG_2$ & $iG_3$ & $iG_4$ & $iG_5$ & $iG_6$ & $iG_7$ & $iG_8$ & $iG_9$  \\
 \hline
 $iG_1$ & $0$ & $0$ & $ 0 $  & $i G_5$ & $-iG_4$ & $iG_7$ & $-iG_6$& $0$ & $0$  \\
 \hline
 $iG_2$ & * & $0 $ & $0 $ & $-iG_5$ & $iG_4$ & $0$ & $0$ & $iG_9$ & $-iG_8$\\
 \hline
 $iG_3$ & * & * & $0$ & $0$ & $0$ & $-iG_7$ & $iG_6$ & $-iG_9$ & $iG_8$ \\
 \hline
 $iG_4$& *& *& *& $0$ & $i(G_1-G_2)$ & $\frac{i}{\sqrt{2}}G_9$ & $\frac{-i}{\sqrt{2}} G_8$ & $\frac{i}{\sqrt{2}} G_7$ & $\frac{-i}{\sqrt{2}} G_6$ \\
 \hline
 $iG_5$ & * & * & * & * & $ 0 $ & $\frac{i}{\sqrt{2}} G_8$ & $\frac{i}{\sqrt{2}} G_9$ & $\frac{-i}{\sqrt{2}}G_6$ & $\frac{-i}{\sqrt{2}} G_7$\\
 \hline 
 $iG_6$ & * & * & * & * & * & $0$ & $i(G_1-G_3)$ & $\frac{i}{\sqrt{2}} G_5$ & $\frac{i}{\sqrt{2}} G_4$ \\
 \hline
 $iG_7$ & * & * & * & * & * & * & 0 &  $\frac{-i}{\sqrt{2}} G_4$ & $\frac{i}{\sqrt{2}} G_5 $ \\ 
 \hline
 $iG_8$ & * & * & * & * & * & * & * & $0$       & $i(G_2-G_3)$ \\
 \hline 
 $iG_9$ & * & * & * & * & * & * & * & * & 0\\
 \hline
 \end{tabular} 
 \end{center}
 \caption{Commutation relations in $u(3)$}\label{Tavola1}
\end{table}

\end{widetext}

We can expand  a matrix $\rho$ in $iu(3)$, as $\rho=\sum_{j=1}^9\rho_jG_j$ and define the real vector $\vec \rho=[\rho_1,...,\rho_9]^T$,  the {\it vectorization} of $\rho$. Analogously,  we can expand a matrix $U$ in $u(3)$ as $U=\sum_{j=1}^9u_jiG_j$ with the vectorization of $U$, given by $\vec u=[u_1,...,u_9]^T$.

Without loss of generality, we can eliminate the diagonal (drift) term $H_0$ in (\ref{LK}) by a time varying coordinate transformation $\rho \rightarrow e^{-H_0t} \rho e^{H_0 t}$ to get a driftless system (a passage to the {\it interaction picture} (see, e.g., \cite{Sakurai})).   This does not affect the magnitude of the control and does not affect the population in the energy level $E_1$ and therefore it does not change the cost we are going to minimize. Therefore from this point onward, we shall consider the dynamics for a Lambda system, 
\begin{equation}\label{refsys}
\dot \rho=\sum_{j=4}^7 [iG_j, \rho] u_j. 
\end{equation}
The problem is to find control functions $u_{4,5,6,7}=u_{4,5,6,7}(t)$, defined in $[0,T]$ to drive the state from an initial condition of the form $\rho_0:=\begin{pmatrix} 0 & 0 \cr 0 & \hat \rho_0 \end{pmatrix}$ to a final condition of the form  $\rho_1:=\begin{pmatrix} 0 & 0 \cr 0 & \hat \rho_1 \end{pmatrix}$, with $\hat \rho_0$ and $\hat \rho_1$,  $2 \times 2$,  density matrices (Hermitian, positive semidefinite and with trace $1$). We want to do this by minimizing the cost functional 
\begin{equation}\label{occupacost}
J=J(T,\gamma_0)=\int_0^T \frac{1}{2} \left( \sum_{j=4}^7 u_j^2 (t) \right)+\gamma_0 \langle G_1, \rho(t)\rangle dt, 
\end{equation}
for some $\gamma_0\geq 0$. The first term  of this cost is the {\it energy} of the control, the second term is the average {\it occupancy} in the highest energy level.   For $\gamma_0=0$, we obtain a minimum energy problem. These types of problems for three level quantum systems have been studied with tools of sub-Riemannian geometry for the `lifted' control problem on the Lie group $SU(3)$ (cf. \cite{conBenj}, \cite{BoscaKP}, \cite{conben}). As $\gamma_0$ increases we penalize the occupancy more and more.

By defining  $U \in su(3)$ as $U=\sum_{j=4}^7 iG_j u_j$, system (\ref{refsys}) with the initial condition $\rho_0$ can be written as 
\begin{equation}\label{refsys2}
\dot \rho=[U,\rho], \qquad \rho(0)=\rho_0,
\end{equation}
and the cost $J$ in (\ref{occupacost}) is written as 
\begin{align}\label{occupacost2}
J&=J(T,\gamma_0)\\
&=\int_0^T\frac{1}{2} \langle U, U \rangle+\gamma_0 \langle G_1, \rho(t)\rangle dt \notag\\
&=\int_0^T\frac{1}{2} \| U \|^2 +\gamma_0 \langle G_1, \rho(t)\rangle dt \notag
\end{align}
We notice that if $u=u(t)$, $t \in [0,T]$, drives $\rho_1$ to $\rho_0$ with cost $J$, and trajectory $\rho=\rho(t)$,  then $-u(T-t)$ drives $\rho_0$ to $\rho_1$ with the same cost $J$ and trajectory $\rho(T-t)$. {Therefore we can always exchange the initial and final conditions in any given problem}.

We also notice that if $K$ is a unitary matrix of the form $K:=\begin{pmatrix} e^{i\phi} & 0 \cr 0 & \hat K \end{pmatrix}$ with $\hat K$ a $2 \times 2$ unitary and $\rho=\rho(t)$ is a solution of (\ref{refsys}) with initial condition $\rho_0$ and for a control function $u_{4,5,6,7}$,  then $K\rho K^\dagger$ is also a solution for initial condition $K \rho_0 K^\dagger$  and a control which has the same norm (at any time) as  the original control and a cost,  (\ref{occupacost}), (\ref{occupacost2}),  which is unchanged. Therefore there is no loss of generality in assuming that the initial condition $\rho_0$ is diagonal, that is, of the form 
\begin{equation}\label{rho0}
\rho_0:=\begin{pmatrix} 0 & 0 & 0 \cr 0 & a & 0 \cr 0 & 0 & 1-a \end{pmatrix},
\end{equation}
with $a > 1-a$, while, for the final condition we shall consider a (general) matrix of the form 
\begin{equation}\label{rho1}
\rho_1=\begin{pmatrix} 0 & 0 & 0 \cr 0 & b & \alpha \cr 0 & \alpha^* & 1-b \end{pmatrix}. 
\end{equation}

From a differential geometry viewpoint,   the state $\rho$ of our  problem belongs to the 6-dimensional manifold $M$ given by the Hermitian matrices which are isospectral to $\rho_0$, that is, 
\begin{equation}\label{maniM}
M:=\{ \rho=X\rho_0 X^\dagger  \, | \, X \in SU(3)  \}. 
\end{equation}
The tangent space at $\rho$, $T_\rho M$ can be identified with 
$$
T_\rho M:=[\rho, su(3)], 
$$
and by using the Frobenius inner product $\langle \cdot, \cdot \rangle$, we can also identify the cotangent space at $\rho$, $T^*_\rho M$,  with $ [\rho, su(3)]$ with the pairing between tangent and cotangent vectors  given by the inner product itself. In particular, when we talk about an element in the cotangent space $T^*_\rho M$, we shall write it as $\sigma \in [\rho, su(3)]$ with the understanding that we mean the linear map $T_\rho M \rightarrow \mathbb{R}$ given by $\langle \sigma, \cdot \rangle$. Write 
$ \sigma \in  T^*_\rho M$ as $\sigma=\sum_{j=1}^9a_j G_j$. If $\sigma \in T^*_{\rho_0} M$, with $\rho_0$ in (\ref{rho0}), we have $a_1=a_2=a_3=0$. If $\sigma \in T^*_{\rho_1}M$, with $\rho_1$ in (\ref{rho1}) then, by calculating $[\rho_1, su(3)]$ using the commutation relations in the above table \ref{Tavola1}, we obtain that $\sigma$ has the expansion\footnote{Write $\rho$ as $\rho=bG_2+(1-b)G_3+\sqrt(2)cG_8+\sqrt{2} c G_9$, where $\alpha:=c+id$, and then use the commutation relations   table \ref{Tavola1}. We  obtain that anything in $[\rho, su(3)]$ must be the linear combination of $G_4,$ $G_5$, $G_6$, and $G_7$ and an element in the two dimensional space ${\cal S}$. A preservation of dimension argument (that is $X[\rho_0, su(3)]X^{\dagger}=[X\rho_0 X^\dagger, su(3)]=[\rho_1,su(3)]$, gives the result. }
$$
\sigma=\sum_{j=4}^7 a_j G_j+ A+B,
$$
where $A$ and $B$ are two matrices in the two dimensional space ${\cal S}$ 
spanned by $\{ dG_8-cG_9, (2b-1)G_9+\sqrt{2}d(G_3-G_2), (1-2b)G_8-\sqrt{2}c(G_3-G_2) \}$, 
where $\alpha:=c+id$ is the off diagonal element of $\rho_1$.  Notice in particular that if $\rho_1$ is real (a case that will interest us in the following) $d=0$ and $\sigma \in T^*M$ has the form 
$$
\sigma=\sum_{j=4}^7 a_j G_j+a_8\left( (1-2b) G_8+\sqrt{2}c(G_3-G_2)\right)+a_9 G_9. 
$$

Let us  consider now   the Lie group ${\bf G}:=S^1 \times S^1$ and its representation given by $ \{ K \in SU(3) \, | \, K=\texttt{diag} ( e^{i\phi}, e^{i\psi}, e^{-i(\phi+\psi)} )\}$ acting on $iu(3)$ as,    $K \in {\bf G}$, $\rho \in iu(3)$,  $\rho \rightarrow K \rho K^\dagger$. If $(U,\rho)$ is a pair of control-trajectory from $\rho_0$ to $\rho_1$  for (\ref{refsys2}) with cost $J$, then $(KUK^\dagger, K\rho K^\dagger)$ is an admissible pair from $K\rho_0 K^\dagger=\rho_0$ to $K\rho_1 K^\dagger$, with the same cost $J$.\footnote{Notice that such a subgroup leaves the 
{\it diagonal} initial condition $\rho_0$ in  (\ref{rho0}) unchanged.}  Therefore the minimum cost of driving $\rho_0$ to $\rho_1$ is the same as the minimum cost of driving $\rho_0$ to $K \rho_1 K^\dagger$, for any $K \in {\bf G}$. Furthermore if we find the optimal control $U$ to drive $\rho_0$ to $\rho_1$, we also have the optimal control $KUK^\dagger$ to drive $\rho_0$ to $K\rho_1K^\dagger$. We can therefore set up the problem as finding the minimum cost $J$ in (\ref{occupacost}), (\ref{occupacost2}) to reach the {\it $1$-dimensional submanifold of $M$}, (cf, (\ref{rho1}), (\ref{maniM}))
\begin{align}\label{subman}
M_1&:=\{ K \rho_1 K^\dagger \, | \, K \in {\bf G} \}\\
&=\left\{ \begin{pmatrix} 0 & 0 & 0 \cr 0 & b & \alpha e^{i\psi} \cr 0 & \alpha^* e^{-i \psi} & 1-b \end{pmatrix} \, | \, \psi \in \mathbb{R} \right\}, \notag
\end{align}
and it does not matter which final condition we achieve in $M_1$. They all give the same cost. Therefore, without loss of generality, we shall assume that the final condition is {\it real}, that is 
\begin{equation}\label{realrho}
\rho_1=\begin{pmatrix} 0 & 0 & 0 \cr 0 & b & N \cr 0 & N & 1-b \end{pmatrix}, \quad  b(1-b)-N^2=a(1-a), 
\end{equation}
where the last equality comes from the fact that $\rho_0$ and $\rho_1$ are isospectral. 
{ We can also choose an arbitrary {\it sign} for $N$ and it will be convenient to choose $N\leq 0$. Furthermore, we shall assume $N\not=0$. The case $N=0$ has either a trivial solution $u\equiv0$ in the case where initial and final condition agree or it is   a {\it population  inversion} problem which has the same solution as the population transfer problem for pure states considered in \cite{conYasemin} }. 

Finally, we shall set  $T=1$ in (\ref{occupacost}) and (\ref{occupacost2}). This is also without loss of generality, since if $(U=U(t), \rho=\rho(t))$ is an admissible  pair control-trajectory in $[0,T]$ (driving $\rho_0$ to $\rho_1$), $(TU(Tt),\rho(Tt))$ is an admissible pair.\footnote{In the following by {\it admissible control and trajectory}, we shall mean a measurable  and bounded control and the corresponding trajectory $\rho$ solution of (\ref{refsys}) with desired final and initial conditions (that is, control driving the state $\rho$  from $\rho_0$ to $\rho_1$) in $[0,T]$.} A standard change of variable argument shows that the costs in the interval $[0,T]$ and $[0,1]$, are related by 
$$
J(1,\gamma_0)=TJ\left(T,\frac{\gamma_0}{T^2}\right). 
$$
Therefore, since we have not put a restriction on $\gamma_0 \geq 0$, once we solve the optimal control problem on $[0,1]$ we can obtain the solution on $[0,T]$ for a scaled $\gamma_0$ and vice versa.

 \section{Background in geometric optimal control theory; The Pontryagin Maximum Principle}\label{PMPsec}

In geometric control theory, the standard necessary first order conditions for optimality are given by the {\it Pontryagin Maximum Principle (PMP)}  (see, e.g., \cite{Agrachev}, \cite{Jurdjevic} and \cite{BoscaRev}, \cite{Mybook} for application to quantum systems). We state below the theorem for a {\it Lagrange problem}, that is a cost that has just an integral part (and no cost on the final condition\footnote{For a more general  formulation of the theorem in the $\mathbb{R}^n$ case  we refer to \cite{Flerish}.}), as in our case, on a finite horizon $[0,T]$.

Let us consider a general system (a family of vector fields parametrized by control functions $u $) for  $x$ in a manifold $M$, 
\begin{equation}\label{gensys}
\dot x=f(x,u), \qquad x(0)=x_0, 
\end{equation}
with cost to be minimized 
\begin{equation}\label{costtoB}
J=\int_0^T L(x,u)dt, 
\end{equation}
and $u$ a (vector) control function measurable and bounded in the interval $[0,T]$ with values in a set $\tilde U$, and $L$ is a smooth function defined on $M \times \tilde U$. 
The final state $x(T)$ is constrained to be in a submanifold of $M$, 
${\cal S} \subset M$. In this setting the {\it Pontryagin Maximum Principle} is given by the following theorem. 
\begin{theorem}\label{PMP}
Assume $(x^*,u^*)$ is an optimal pair for the above problem. Then there exists a function $\lambda=\lambda(t)$ with values in the cotangent bundle of $M$, $T^*M$,  and a constant $\mu_0 \leq 0$ not both (identically) zero  such that, defined the function,\footnote{Notice the abuse of notation here. Since $\lambda$ is an element of the cotangent bundle it already contains the information on its base element $x$. The function $\hat H$ is in fact a function of $\lambda$ and $u$ only, that is, defined on $TM \times \tilde U$.}  
\begin{equation}\label{PMPHamiltonian}
\hat H=\hat H(\lambda,x,u)=\lambda^Tf(x,u)+\mu_0L(x,u),\footnote{Here we denote by $\lambda^Tf$ the pairing between an element of the cotangent space, $\lambda$,  and an element of the tangent space $f$, reminiscent of what happens in $\mathbb{R}^n$.}    
\end{equation}
we have 
\begin{enumerate}
\item For almost every $t \in [0,T]$, 
\begin{equation}\label{maximizcond}
 H(\lambda(t),x^*(t),u^*(t)) \geq  H(\lambda(t),x^*(t),v), \quad \forall v \in \tilde U.
\end{equation}

\item $\lambda$ satisfies the differential equation (to be coupled with (\ref{gensys}))
\begin{align}\label{equala}
\dot \lambda^T&=-\lambda^Tf_x(x^*,u^*)-\mu_0L_x(x^*,u^*)\\
&=-\frac{\partial \hat H}{\partial x} (\lambda, x^*,u^*). \notag
\end{align}

\item $H(\lambda(t), x^*(t), u^*(t))$ is constant over all the interval $[0,T]$. 

\item (Transversality Conditions) Denote by $T_{x(T)} {\cal S}$ the tangent space of ${\cal S}$ at the final point $x(T)$. Then 
 \begin{equation}\label{transversalitycond}
\lambda^T(T) T_{x(T)} {\cal S} =0. 
\end{equation}
\end{enumerate}
\end{theorem}
Pairs $(x^*,u^*)$ which satisfy the above conditions are called {\it extremals}. They are candidate to be optimal. If they satisfy the above conditions with $\mu_0=0$ they are called {\it abnormal extremals}. These are pairs control-trajectory which satisfy the conditions of the theorem { independently of the function $L$ in the cost (\ref{costtoB})}. Otherwise, they are called {\it  normal extremals}. In that case $\mu_0$ can be normalized to be equal to $-1$. The vector $\lambda$ is called the {\it costate}.\footnote{Notice that if $\mu_0=0$ and $\lambda(0)=0$, $\lambda(t)$ is zero for every $t$.} We shall call the function $\hat H$ in (\ref{PMPHamiltonian}), the {\it control Hamiltonian}.  

Most of the theory of optimal control, such as the above theorem,  is developed for dynamical equations and costs, which are described in terms of {\it real} variables/functions. Therefore, it is convenient to see equation (\ref{LK}) and the costs above introduced in terms of real variables/functions. Let us study how the conditions from the Pontryagin Maximum Principle specialize for the bilinear control system of the form (\ref{LK}).

The dynamical equation (\ref{LK}) in terms of real variables takes the form
 \begin{equation}\label{formareal}
 \frac{d}{dt} \vec \rho=ad_{H_0} \vec \rho+\sum_{j=1}^m u_j ad_{H_j}  \vec \rho  
 \end{equation} 
 where $ad_{H_j}$, $j=0,1,...,m$,  are the matrices corresponding (in the given basis (such as (\ref{base})) to the operators of commutators with $H_0,H_1,...,H_m$. The costate  equations (\ref{equala}), using (\ref{formareal}),  become 
$$
\frac{d}{dt} \lambda^T=-\lambda^T \left( ad_{H_0}+\sum_{j=1}^m ad_{H_j} u_j\right)-\mu_0L_{\vec \rho} (\vec \rho, u). 
$$
The matrices $ad_{H_0}, ad_{H_1},...,ad_{H_m}$ are skew symmetric (if we assume, as we do here, that $\rho$ is expanded with respect to an orthonormal basis). Therefore we can write 
\begin{equation}\label{precurs}
\frac{d}{dt} \lambda= \left( ad_{H_0}+\sum_{j=1}^m ad_{H_j} u_j \right) \lambda -\mu_0L_{\vec \rho}^T (\vec \rho, u). 
\end{equation}

\section{Differential equations for the optimal candidates}\label{SysT}

Using the PMP discussed in the previous section, we shall now derive a system of differential equations which has to be satisfied by the optimal control-trajectory pairs for our problem. We shall go back and forth between the matrix form of the variables and the vectorized form.

Let us  rewrite equation (\ref{precurs})  in matrix form and let $\sigma$ be the Hermitian (costate) zero-trace matrix associated with $\lambda$, that is, $\lambda=\vec \sigma$. The term $L^T_{\vec \rho}(\vec \rho, u)$ in (\ref{precurs}) may have different forms. In the case of an occupancy cost where the part of $L$ depending on $\vec \rho$ is $\gamma_0 \vec k^T \vec \rho= \gamma_0 \langle K, \rho\rangle$, it is $\gamma_0 \vec k^T$ which becomes $\vec k$ after transposition which corresponds to a matrix $K$ in the matrix notation. Thus, in this case,  the adjoint equations in the matrix notation are written as
$
\frac{d}{dt} \sigma= [H_0, \sigma]+\sum_{j=1}^m [H_j, \sigma] u_j
-\mu_0 \gamma_0 K, 
$
which, specializing to the three-level Lambda system governed by the Liouville-von Neumann equation (\ref{refsys}) with cost (\ref{occupacost}) becomes 
\begin{equation}\label{costaeq}
    \dot{\sigma}=\left[\sum_{j=4}^7 iG_ju_j,\sigma\right]-\mu_0\gamma_0 G_1.
\end{equation}

The control Hamiltonian is
\begin{equation}\label{controham}
    \hat H=\langle\sigma,\left[\sum_{j=4}^7 iG_ju_j,\rho\right]\rangle+ \mu_0 \left( \gamma_0\langle G_1,\rho \rangle+\frac{1}{2}\sum_{j=4}^7 u_j^2 \right).
\end{equation}

\vspace{0.25cm}

\noindent Let us define the variables 
\begin{equation}\label{defihj}
    h_j:=\langle\sigma,[iG_j,\rho]\rangle, \qquad j=1,...,9, 
\end{equation} with which the control Hamiltonian $\hat H$ in (\ref{controham}) can be written as 
\begin{equation}\label{contrham}
\hat H=\sum_{j=4}^7 h_j u_j+\mu_0 \left( \gamma_0 \langle G_1,\rho\rangle+\frac{1}{2} \sum_{j=4}^7 u_j^2 \right).
\end{equation}
Differentiating  $h_j$, in (\ref{defihj}) and using (\ref{refsys}) and (\ref{costaeq}), along with the definition $U:=\sum_{k=4}^7 iG_k u_k$, we obtain,
\begin{align}
    \dot h_j &= \langle \dot\sigma,[iG_j,\rho]\rangle+\langle \sigma, [iG_j,\dot\rho]\rangle,\label{auxiliaryE} \\
    &= \langle [U,\sigma], [iG_j,\rho]\rangle -\mu_0 \gamma_0 \langle G_1, [iG_j,\rho]\rangle + \langle [U,\rho],[\sigma, iG_j]\rangle  \nonumber\\
    &= -\mu_0 \gamma_0 \langle G_1, [iG_j,\rho]\rangle -\langle \sigma, [U, [iG_j, \rho]]\rangle - \langle \sigma, [iG_j, [\rho, U]]\rangle \nonumber \\
    &=-\mu_0 \gamma_0 \langle \rho, [G_1,i G_j]\rangle+\langle \sigma, [\rho, [U, iG_j]]\rangle. \nonumber
\end{align}
Here we have used the property of the Frobenius inner product $\langle A, [B,C]\rangle=\langle B, [C,A] \rangle$ and the Jacobi identity $[A,[B,C]]+[B,[C,A]]+[C,[A,B]]=0$. Consider now the cases $j=1,2,3,8,9$ in (\ref{auxiliaryE}). The first term on the right hand side is zero since $[G_1,iG_j]=0$. As for the second term, we calculate 
$
[U, iG_j]=\sum_{k=4}^7 [iG_k, iG_j]u_k. 
$
Consider, as an illustration, the case $j=1$ (the other cases are similar). Using the commutators in Table \ref{Tavola1} we obtain 
$$
\sum_{k=4}^7 [iG_k, iG_1]u_k=-iG_5 u_4+iG_4u_5-iG_7 u_6 +iG_6 u_7, 
$$
which placed into (\ref{auxiliaryE}) with the definitions (\ref{defihj})  gives 
$$
\dot h_1=h_5 u_4 -h_4 u_5+h_7 u_6 -h_6 u_7. 
$$
Now, there are two possibilities for the control Hamiltonian (\ref{contrham}): 1) $\mu_0=0$, in which case the extremal is abnormal. In this case the maximization condition (\ref{maximizcond}) of Theorem \ref{PMP} gives $h_4\equiv h_5\equiv h_6\equiv h_7\equiv 0$. 2) $\mu_0=-1$, in which case the maximization condition  gives 
$$
u_j=h_j,\qquad j=4,5,6,7. 
$$
In both cases, we have $\dot h_1\equiv 0$. Analogously we find that $\dot h_2\equiv \dot h_3 \equiv \dot h_8 \equiv \dot h_9 \equiv 0$. Therefore $h_j$, for $j=1,2,3,8,9$ are constants. Information on the values of these constants can be obtained using the initial and final conditions and the transversality conditions (\ref{transversalitycond}). In particular, calculating $h_{1,2,3}$ at time $t=0$,  using the definition (\ref{defihj}) and the fact that $\rho_0$ in (\ref{rho0}) commutes with $G_{1,2,3}$, we obtain $h_1=h_2=h_3=0$.  To obtain the condition for $h_8$, we recall that, because of the symmetry by the Lie group ${\bf G}$ defined in the paragraph leading to formula  (\ref{subman}),  
 an optimal control to reach $\rho_1$ in (\ref{realrho}) is also an optimal control to reach the submanifold   $M_1$ in (\ref{subman}). Therefore from the transversality condition (\ref{transversalitycond}), we have for any real $k$ and $N$ defined in (\ref{realrho}), 
$$
\frac{d}{dt}|_{t=0} \langle \sigma(T), e^{iDt} \rho_1 e^{-iDt} \rangle= kN\langle \sigma(T),G_9\rangle=0, 
$$
which implies since we have assumed $N\not=0$, $\langle \sigma(T), G_9 \rangle =0$. Now calculate $h_8:=\langle \sigma, [iG_8, \rho]\rangle$, at time $T$, using from (\ref{realrho}), $\rho=bG_2+\sqrt{2}NG_8+(1-b)G_3$ , and the commutation relations in Table \ref{Tavola1}, 
\begin{align}
h_8&=\langle \sigma, [iG_8, \rho] \rangle\\
&=b \langle \sigma, [iG_8, G_2]\rangle + (1-b) \langle \sigma, [iG_8, G_3]\rangle \notag\\
&=(1-2b)\langle \sigma, G_9\rangle=0.  \notag
\end{align}

We now derive the equations for $h_{4,5,6,7}$, keeping in mind that we have just proved $h_1=h_2=h_3=h_8=0$. Using (\ref{Tavola1}), we obtain, 

\begin{widetext}
    \begin{align}
   [iG_4,U]= \left[iG_4,\sum_{j=4}^7 iG_j u_j\right]&=i[(G_2-G_1)u_5+\frac{1}{\sqrt{2}}G_8u_7-\frac{1}{\sqrt{2}}G_9u_6], \nonumber \\
   [iG_5, U]=  \left[iG_5,\sum_{j=4}^7 iG_j u_j\right]&=i[(G_1-G_2)u_4-\frac{1}{\sqrt{2}}G_8u_6-\frac{1}{\sqrt{2}}G_9u_7], \nonumber \\
    [iG_6, U]=\left[iG_6,\sum_{j=4}^7 iG_j u_j\right]&=i[(G_3-G_1)u_7+\frac{1}{\sqrt{2}}G_8u_5+\frac{1}{\sqrt{2}}G_9u_4], \nonumber \\
    [iG_7,U]= \left[iG_7,\sum_{j=4}^7 iG_j u_j\right]&=i[(G_1-G_3)u_6-\frac{1}{\sqrt{2}}G_8u_4+\frac{1}{\sqrt{2}}G_9u_5], \nonumber 
\end{align}
\end{widetext}

and replacing in (\ref{auxiliaryE}) using again the commutation relations in Table \ref{Tavola1}, we obtain,  

\begin{widetext}
    \begin{align}
   \dot{h}_4 &= -\mu_0 \gamma_0 \langle \rho, G_5 \rangle + h_2 u_5 -h_1 u_5 +\frac{h_8}{\sqrt{2}} u_7 -\frac{h_9}{\sqrt{2}} u_6= 
   -\mu_0 \gamma_0 \langle \rho, G_5 \rangle -\frac{h_9}{\sqrt{2}} u_6, \label{Auxi}\\
   \dot{h}_5 &= \mu_0 \gamma_0 \langle \rho, G_4\rangle +h_1 u_4-h_2 u_4 -\frac{h_8}{\sqrt{2}}u_6-\frac{h_9}{\sqrt{2}}u_7= 
   \mu_0 \gamma_0 \langle \rho, G_4\rangle -\frac{h_9}{\sqrt{2}}u_7,  \nonumber \\
   \dot{h}_6 &=-\mu_0 \gamma_0 \langle \rho, G_7\rangle + h_3 u_7-h_1 u_7   + \frac{h_8}{\sqrt{2}}u_5+\frac{h_9}{\sqrt{2}}u_4=-\mu_0 \gamma_0 \langle \rho, G_7\rangle 
   +\frac{h_9}{\sqrt{2}}u_4, \nonumber  \\
   \dot{h}_7 &= \mu_0 \gamma_0 \langle \rho, G_6 \rangle + h_1 u_6-h_3 u_6 -\frac{h_8}{\sqrt{2}} u_4 +\frac{h_9}{\sqrt{2}} u_5=
  \mu_0 \gamma_0 \langle \rho, G_6 \rangle +\frac{h_9}{\sqrt{2}} u_5 . \nonumber 
\end{align}
\end{widetext}

From these equations, we now prove the following theorem. 

\begin{theorem}\label{normalita}
The optimal control problem for system (\ref{refsys}) and cost (\ref{occupacost}) and intial and final conditions (\ref{rho0}) (\ref{realrho}) admits only {\em normal} extremals.
\end{theorem}
\begin{proof}
Assume by contradiction $\mu_0=0$ so that we can rewrite equations (\ref{Auxi}) in matrix form as 
 \begin{equation}\label{abauxmat}
 \begin{pmatrix}
 \dot h_4 \\
 \dot h_5 \\
 \dot h_6 \\
 \dot h_7
 \end{pmatrix}=
      \frac{1}{\sqrt{2}}\begin{pmatrix}
      0 & 0 & -h_9 & 0 \\
      0 & 0 & 0 & -h_9  \\
      h_9 & 0 & 0 & 0  \\
      0 & h_9 & 0 & 0 
   \end{pmatrix}
   \begin{pmatrix}
      u_4 \\
       u_5 \\
        u_6 \\
         u_7
\end{pmatrix}. 
\end{equation}
 Applying the maximization condition to the control Hamiltonian (\ref{contrham}) with $\mu_0=0$ leads to $h_4\equiv h_5\equiv h_6\equiv h_7 \equiv 0$. So equation (\ref{abauxmat}) gives $u_4\equiv u_5 \equiv u_6 \equiv u_7 \equiv 0$ unless $h_9=0$. The identically zero control is clearly not an admissible extremal because it will correspond to a constant trajectory. Therefore, necessarily we must have $h_9=0$.  However, in this case  $h_k=0$ for all $k\in\{1,\dots,9\}$. This implies that $\langle\sigma,\left[u(3),\rho\right]\rangle=0$.  Since $\left[u(3),\rho\right]$ is  the tangent space at $\rho$ of the state manifold,  this means that the cotangent vector $\sigma$ is zero on the whole tangent space, which contradicts Pontryagin's Maximum Principle (the pair $\mu_0,\sigma$ cannot be zero).

 \end{proof}

Since abnormal extremals do not exist, we can set $\mu_0=-1$ in (\ref{contrham}). Carrying out the maximization with respect to $u$ of condition (\ref{maximizcond}), we get $u_j=h_j$, for $j=4,5,6,7$, and the differential system (\ref{Auxi}) for $h_4,...,h_7$ becomes a differential system for $u_4,...,u_7$. Let us now write such a differential system in matrix form. By expanding $U$ as $U:=\sum_{j=4}^7 iG_j u_j$, taking the derivative of $U$ and using (\ref{abauxmat}) with the $h_j$'s replaced by the corresponding $u_j$'s,  we get,  with $\tilde h_9=\frac{h_9}{\sqrt{2}}$, 
$$
\dot U=\tilde h_9C+i\gamma_0 B, 
$$
where $C:=-iG_4u_6-iG_5u_7+iG_6u_4+iG_7u_5$ and $B:=G_4\rho_5-G_5\rho_4+G_6\rho_7-G_7 \rho_6)$. A direct calculation, using the commutators in Table \ref{Tavola1} shows that $C=\sqrt{2}[iG_9,U]$. Therefore, using the definition of $\tilde h_9$, we have 
$$
\dot U=h_9[iG_9, U]+i\gamma_0B. 
$$
Consider now $B$. It is convenient to decompose $\rho$ as $\rho=\rho_P+\rho_K$ where $\rho_P$ ($\rho_K$) is the (block) anti-diagonal (diagonal) part of $\rho$. In particular $\rho_P:=\sum_{j=4}^7 \rho_jG_j$. Using again the commutators in Table \ref{Tavola1}, with this expression, we get $[iG_1, \rho_P]=-B$. Therefore we get 
$$
\dot U=h_9[iG_9,U]-i \gamma_0 [iG_1, \rho_P]. 
$$ 
However, notice that $G_1$ commutes with $\rho_K$. Therefore $ [iG_1, \rho_P]= [iG_1, \rho]$. Summarizing, and combining with (\ref{refsys2})  we can write the equations satisfied by the optimal control and trajectory as 
\begin{eqnarray}
\dot \rho=[U,\rho], \qquad \rho(0)=\rho_0 \label{aux1}\\
\dot U=h_9[iG_9, U] +\gamma_0[G_1,\rho]. \qquad U(0)=P\label{aux2}
\end{eqnarray}
The matrix $P=U(0)$ is a matrix of the form 
\begin{equation}\label{matrixP}
P=\begin{pmatrix}0 & \alpha & \beta \cr -\alpha^* & 0 & 0 \cr -\beta^* & 0 & 0 \end{pmatrix}, 
\end{equation}
which gives the initial condition of the control. The parameters $\alpha$ and $\beta$ need to be tuned. We summarize the discussion in the following theorem where we also record the form of the control Hamoltonian $\hat H$ which follows from using $h_j=u_j$ in (\ref{contrham}) with $\mu_0=-1$. 

\begin{theorem}\label{summary}
Assume $(\rho,U)$ with $U=\sum_{j=4}^7u_jiG_j$ is an optimal pair. Then it  satisfies  (\ref{aux1}) (\ref{aux2}) for some value of the matrix $P$ in (\ref{matrixP})  and the parameter $h_9$ in (\ref{aux2}). The control Hamiltonian $\hat H$ is constant along the optimal trajectory, and it is given by 
\begin{equation}\label{cHcon}
\hat H=\frac{1}{2} \|U(t)\|^2+\gamma_0 \langle G_1, \rho(t)\rangle \equiv \frac{1}{2}\|P\|^2, 
\end{equation}
\end{theorem}

\vspace{0.25cm}

We have therefore reduced the optimal control problem to the following:

\vspace{0.25cm}

\noindent {\bf Problem:} {\it Find a matrix $P$ of the form (\ref{matrixP}) and a real constant $h_9$ such that 

\begin{enumerate}
\item The solution of (\ref{aux1}),(\ref{aux2}) attains the value $\rho_1$ in (\ref{realrho}) at time $T$ (recall we can take $T=1$). 

\item The cost (\ref{occupacost2}) is minimized. 

\end{enumerate}}

\section{Reduction of the number of parameters and auxiliary results}\label{Red}

The {\bf Problem} stated at the end of the previous section involves a search over $\mathbb{C}^2 \times \mathbb{R}$, an infinite number of parameters.  In order to reduce the search to a compact set and to reduce possible redundancies, we shall develop some further auxiliary results. We start by studying the dependence of the optimal cost on the value $\gamma_0$ which will give us   bounds on the value of the optimal cost for a certain $\gamma_0$ in terms of the cost for a different $\gamma_0$. From this we can derive bounds for $\|P\|=\|U(0)\|$. 

\subsection{Dependence of the optimal cost on $\gamma_0$}. 

Consider two values of $\gamma_0$, $\gamma_1 < \gamma_2$, and denote by $J_{\gamma_1}$ ($J_{\gamma_2}$) the optimal cost for $\gamma_1$ ($\gamma_2$) and by $(U_1,\tilde \rho_1)$   
($(U_2,\tilde \rho_2)$) the corresponding pair of optimal control-trajectory. Thus we have for the optimal costs,  
$$
J_{\gamma_{1,2}}=\int_0^1 \frac{1}{2} \|U_{1,2} (t) \|^2+\gamma_{1,2} \langle G_1, \tilde \rho_{1,2}(t)\rangle dt. 
$$
The following inequalities hold:
\begin{widetext}
    \begin{equation}\label{inc}
J_{\gamma_1}\leq \int_0^1\frac{1}{2} \|U_2\|^2 +\gamma_1 \langle G_1,\tilde \rho_2\rangle dt \leq J_{\gamma_2}\leq \int_0^1\frac{1}{2} \|U_1\|^2+\gamma_2\langle G_1, \tilde \rho_1\rangle dt=
J_{\gamma_1}+(\gamma_2-\gamma_1)\int_0^1 \langle G_1,\tilde \rho_1\rangle dt. 
\end{equation}
\end{widetext}

The first inequality is due to the optimality of $(U_1,\tilde \rho_1$) for $\gamma_1$. The second one follows from $\gamma_2> \gamma_1$. The third one follows from the optimality of  $(U_2,\tilde \rho_2$) for $\gamma_2$. Notice that since $\int_0^1 \langle G_1,\tilde \rho_1\rangle$ is upper bounded by the maximum eigenvalue of $\tilde \rho_1$ which is the same as the maximum eigenvalue of $\rho_0$ in (\ref{rho0}), that is, $a$.  We have 
\begin{equation}\label{crescita}
J_{\gamma_1}\leq J_{\gamma_2} \leq J_{\gamma_1}+a(\gamma_2-\gamma_1). 
\end{equation}
The following holds. 
\begin{proposition}\label{crescprop}
The optimal cost is nondecreasing with $\gamma_0$ and grows at most linearly according to (\ref{crescita}). 
\end{proposition}
A situation where the optimal cost and the optimal control and parameters can be calculated analytically is when $\gamma_0=0$. In this case there is no penalty on the occupancy and the optimal control problem is a {\it minimum energy steering problem}. The solution of the equations  (\ref{aux1}) (\ref{aux2}) for $\gamma_0=0$ can be written explicitly, defining $A:=h_9 iG_9$, as 
\begin{equation}\label{explicexpre}
U(t)=e^{At} P e^{-At}, \quad e^{At}e^{(-A+P)t}\rho_0 e^{(A-P)t} e^{-At}, 
\end{equation}
as it can be directly verified (cf. \cite{JKP}). In this case the norm of $U$ is constant in time and equal to $\|P\|$. Therefore the cost (\ref{occupacost2}) is $J_0=\frac{1}{2}\|P\|^2$. From the explicit expression (\ref{explicexpre}) one can also obtain the integral $\int_0^1 \langle G_1,  \tilde \rho_1\rangle dt$ to be used in finding the bounds in (\ref{inc}) for a given $\gamma_0$. The problem of finding the parameters  $A$ and $P$ to reach the final condition $\rho_1$ at time $T=1$, i.e., $ e^{A}e^{-A+P}\rho_0 e^{A-P} e^{-A}=\rho_1$, with minimum cost, i.e., minimum $\|P\|$ can be solved in two ways. One can `lift' the problem to the corresponding problem on the Lie transformation group $SU(3)$ and apply the results for minimum energy transfer for this problem which have been well developed in the geometric control theory literature (cf.  \cite{conBenj}, \cite{BoscaKP}, \cite{conben}). In particular, one can characterize the set ${\cal S}$ of matrices $X$ in $SU(3)$ which perform the state transfer $X\rho_0X^\dagger =\rho_1$ and then use the results of \cite{conben}) which give the minimum cost as a function of $X$ to find the minimum in the set ${\cal S}$. Once such a minimum is found,  the results of  \cite{conBenj}, \cite{conben}, also give a method to find the matrices $A$ and $P$. Such a method does not take into account that we already know the form of $A$ which is in \cite{conBenj}, \cite{conben}, assumed to be a general block diagonal matrix. An alternative approach consists of adapting the {\it methods} of  \cite{conBenj}, \cite{conben} to our situation (problem for density matrix transfer with known form for the matrix $A$) rather than applying such results after lifting. We shall give an example of this in the next section where we shall calculate the optimal cost  for a particular state transfer (the more technical part of the treatment is presented in Appendix \ref{ZOC1})

\subsection{Bounds on the norm of the control $\|P\|$ }\label{BNC}

Every upper $B_U$ and lower $B_L$ bound for the cost give an upper and lower bound for the norm of the control $\|U(0)\|^2=\|P\|^2$ at time $t=0$. To see this, assume that, as for example in (\ref{inc}), (\ref{crescita}),  
$$
B_L \leq J_{\gamma_0}:= \int_0^1 \frac{1}{2} \|U(t)\|^2 +\gamma_0 \langle G_1, \rho \rangle dt \leq B_U .
$$
Then using the constancy of the control  Hamiltonian in formula (\ref{cHcon}) we get $\frac{1}{2}\|U(t)\|^2=\frac{1}{2}\|U(0)\|^2+\gamma_0\langle G_1, \rho \rangle$, which gives 
$$
B_L \leq \frac{1}{2} \|U(0)\|^2+\int_0^12 \gamma_0 \langle G_1, \rho \rangle dt \leq B_U. 
$$
This gives, using $\langle G_1, \rho\rangle \leq a$ (in (\ref{rho0})), 
$
B_L-2\gamma_0 a \leq \frac{1}{2} \|U(0)\|^2 \leq B_U.
$
Thus we get
\begin{equation}\label{maximin}
\sqrt{\max\{ 0, {2B_L-4\gamma_0 a}\}} \leq \|P\| \leq \sqrt{2} \sqrt{B_U}. 
\end{equation}

\subsection{Symmetries and elimination of redundancy}\label{SRed}

 Consider $D$ a diagonal matrix of the form $D:=\begin{pmatrix} e^{2i\phi} & 0 & 0 \cr 0 & e^{-i\phi} & 0 \cr 0 & 0 & e^{-i\phi} \end{pmatrix}$. If $(U,\rho)$ is a pair of optimal control and trajectory, then 
 the pair $\left( DUD^\dagger, D\rho D^\dagger \right)$ also has the same initial and final conditions for $\rho$ and the same cost, and therefore it is also optimal.   It satisfies the 
 equations (\ref{aux1}), (\ref{aux2}) but with initial condition $DPD^\dagger$ instead of $P$.  Therefore we can set one of the phases  of the  parameters in $P$,  ($\alpha$, $\beta$) in (\ref{matrixP}). For example we can take $\alpha$ real and nonnegative, without loss of generality, by an appropriate choice of $D$. This gives $u_4(0)=0$. Furthermore, we have the following. 

 \begin{proposition}\label{coniugati}
If $(U,\rho)$ is a pair satisfying equations (\ref{aux1}) (\ref{aux2}) of control-trajectory driving from $\rho_0$ to $\rho_1$, then $(\bar U,\bar \rho)$ is a pair of control-trajectory satisfying equations (\ref{aux1}) (\ref{aux2})  driving from $\bar \rho_0$ to $\bar \rho_1$ with the same cost. 
 \end{proposition}

 \begin{proof}
Take the conjugate of equations (\ref{aux1}) and (\ref{aux2}) and using the fact that $iG_9$ is real it follows that if $(U,\rho)$ satisfy these equations so do $(\bar U,\bar \rho)$. The  initial and final conditions  are the conjugate. Using the expression for the cost (\ref{occupacost2}) it immediately follows that the costs are the same. 
 \end{proof}

Since we have reduced ourselves, without loss of generality,  to $\rho_0$ and $\rho_1$ that are {\it real},  if $U$ is an optimal control so is $\bar U$. This allows us to restrict even further the possible initial conditions for $U(0)$. For instance we will take $u_6(0)\geq 0$. 
 
 Summarizing, there will be no loss of generality in assuming that the matrix $P$ in (\ref{matrixP}) is of the form 
 \begin{equation}\label{newmatrixP}
 P=\begin{pmatrix} 0 & p & b+id \cr -p & 0 & 0 \cr -b +id & 0 & 0 \end{pmatrix}, 
 \end{equation}
 with $b,$ $p,$ $d$ real and $p\geq 0$, $d \geq 0$.

 \subsection{Bounds on the values of $h_9$} 
 
 The bounds on the value of $h_9$ in (\ref{aux2}) are based on the intuitive  idea that as $h_9$ becomes large in (\ref{aux2}),  the second term in the right hand side of the equation will become unimportant and the first term will give for $U$ an highly oscillatory matrix which has no effect on $\rho$. Thus $\rho$ will not reach the desired final value of $\rho_1$ in (\ref{realrho}). This is captured by the 
 following estimates which are reminiscent of the method of averaging  \cite{JimM}.

 We start by introducing  a time varying change of coordinates for $U$ in (\ref{aux1}), (\ref{aux2}) by defining the matrix function $\tilde U=e^{-h_9iG_9t} U e^{h_9iG_9t}$, which gives,  from (\ref{aux1}), (\ref{aux2}),  the differential equations for $\rho$ and $\tilde U$, 
 \begin{eqnarray}
\dot \rho=[e^{ih_9G_9 t}\tilde U e^{-ih_9G_9 t},\rho], \quad \rho(0)=\rho_0,  \label{aux1A}\\
\frac{d}{dt} \tilde U=\gamma_0[G_1,e^{-ih_9G_9t} \rho e^{ih_9G_9t}], \quad \tilde U(0)=P.  \label{aux2A} 
\end{eqnarray}
Let us introduce a time scaling by assuming $h_9>0$ (the case $h_9=0$ will be included in the bound we eventually obtain and the case $h_9<0$ only requires small modifications and leads to the same bound for $|h_9|$). Define  $\tau=|h_9|t$, and 
$\hat U(\tau):=\tilde U\left( \frac{\tau}{h_9} \right)$, $\hat \rho=\rho\left( \frac{\tau}{h_9}\right)$. Also set $\epsilon:=\frac{1}{h_9}$. We obtain from (\ref{aux1A}) and (\ref{aux2A}), 
 \begin{eqnarray}
\frac{d}{dt} \hat \rho=\epsilon [e^{iG_9 \tau}\hat U e^{-G_9 \tau },\hat \rho], \quad \hat \rho(0)=\rho_0, \label{aux1B}\\
\frac{d}{dt} \hat  U=\epsilon \gamma_0[G_1,e^{-iG_9\tau} \hat  \rho e^{iG_9\tau }], \quad \hat U(0)=P,   \label{aux2U} 
\end{eqnarray}
over a time interval $[0, h_9]$.  Integrating (\ref{aux1B}),  
we get $\hat \rho(\tau)-\rho_0=\epsilon \int_0^\tau [e^{iG_9 s} \hat U(s) e^{-iG_9s}, \hat \rho(s)]ds$, 
while, integrating (\ref{aux2U}),   we get $\hat U(s)-P=\epsilon \int_0^s[ \epsilon \gamma_0 [G_1, e^{-iG_9 r}\hat \rho(r) e^{iG_9 r}]dr$. Combining these two expressions, we obtain
\begin{widetext}
    $$
\hat \rho(\tau) -\rho_0=\epsilon \int_0^\tau [e^{iG_9 s} P e^{-iG_9s}, \hat \rho(s)]ds+
\epsilon^2 \gamma_0 \int_0^\tau \int_0^s\left[ e^{iG_9 s}[G_1, e^{-iG_9 r} \hat \rho(r) e^{iG_9r}] e^{-iG_9 s}, \hat \rho(s) \right] dr ds.
$$
\end{widetext}

This, using the fact that $G_1$ commutes with $G_9$,  can be written more compactly as 
\begin{widetext}
    \begin{equation}\label{thetwointegrals}
\hat \rho(\tau)-\rho_0=\epsilon  \int_0^\tau [e^{iG_9 s} P e^{-iG_9s}, \hat \rho(s)]ds+\epsilon^2 \gamma_0 \int_0^\tau \int_0^s \left[ [G_1, e^{iG_9(s-r)} \hat \rho(r) e^{iG_9(r-s)}], \hat \rho(s) \right] dr ds. 
\end{equation}
\end{widetext}

Let us study the two integrals in (\ref{thetwointegrals}) (and their norms) separately. 

\begin{enumerate}
\item {\bf First integral}

We write the first integral on the right hand side of (\ref{thetwointegrals})  as 
$$
\int_0^\tau[e^{iG_9 s} P e^{-iG_9s}, \rho_0] ds+\int_0^\tau[e^{iG_9 s} P e^{-iG_9s}, \hat \rho(s)- \rho_0] ds. 
$$
The first term is a periodic function of $\tau$ with period $2\sqrt{2}\pi$ that we can calculate explicitly. Using $P$ in (\ref{newmatrixP}) and $\rho_0$ in (\ref{rho0}), we can obtain an upper bound for the norm of this integral which is proved in the Appendix \ref{upboun} ($a$ is defined in (\ref{rho0}))
\begin{equation}\label{UB3}
\left\|  \int_0^\tau[e^{iG_9 s} P e^{-iG_9s}, \rho_0] ds  \right\| \leq 2\sqrt{2}a\|P\|. 
\end{equation}
For the second term, we obtain the bound, 
\begin{widetext}
    $$
\left \| \int_0^\tau [e^{iG_9 s} P e^{-iG_9s}, \hat \rho(s)- \rho_0] ds \right\| \leq \int_0^\tau \| [e^{iG_9 s} P e^{-iG_9s}, \hat \rho(s)- \rho_0]\| ds\leq \sqrt{2}\|P\|\int_0^\tau \|\hat \rho(s)-\rho_0\|ds. 
$$
\end{widetext}

Here we used the B\"ottcher and Wenzel bound for the Frobenius norm of a commutator $[A,B]$. i.e., $\|[A,B]\| \leq \sqrt{2} \|A\|\|B\|$ (cf.  \cite{BW} and the references therein).

\item {\bf Second Integral}

We write $\int_0^\tau \int_0^s \left[ [G_1, e^{iG_9(s-r)} \hat \rho(r) e^{iG_9(r-s)}], \hat \rho(s) \right] dr ds. $ as 
\begin{widetext}
\begin{equation*}
    \int_0^\tau \int_0^s \left[ [G_1, e^{iG_9(s-r)} (\hat \rho(r)-\rho_0) e^{iG_9(r-s)}], \hat \rho(s) \right] dr ds +\int_0^\tau \int_0^s \left[ [G_1, e^{iG_9(s-r)}  \rho_0 e^{iG_9(r-s)}], \hat \rho(s) \right] dr ds. 
\end{equation*}
\end{widetext}
 and we notice that the second integral is zero because, since $G_1$ commutes with   $G_9$, the inner commutator is 
 $  e^{iG_9(s-r)}  [G_1, \rho_0] e^{iG_9(r-s)}=0$ since $G_1$ commutes with $\rho_0$ (they are both diagonal). Therefore, we obtain 

 \begin{widetext}
      $$
 \left\| \int_0^\tau \int_0^s \left[ [G_1, e^{iG_9(s-r)} \hat \rho(r) e^{iG_9(r-s)}], \hat \rho(s) \right] dr ds.  \right\| \leq 
 $$
 $$
 \int_0^\tau \int_0^s \left \|  \left[ [G_1, e^{iG_9(s-r)}  (\hat \rho(r)-\rho_0) e^{iG_9(r-s)}], \hat \rho(s) \right] \right\| dr ds 
 \leq \
 $$
 $$
 \int_0^\tau \int_0^s \sqrt{2} \sqrt{a^2+(1-a)^2}    \left\| \left[ G_1, e^{iG_9(s-r)}(\hat \rho(r)-\rho_0)e^{iG_0(r-s)} \right] \right\| dr ds \leq 
 $$
$$
\int_0^\tau \int_0^s 2  \sqrt{a^2+(1-a)^2}   \|\hat \rho(r)-\rho_0\| dr ds=
 2  \sqrt{a^2+(1-a)^2} \int_0^\tau (\tau-r) \|\hat \rho(r) - \rho_0\| dr
$$
 \end{widetext}

Here we used the following: In the second inequality, we used again B\"ottcher-Wenzel inequality and the fact that the Frobenius norm of $\hat \rho(s)$ is equal to the Frobenious norm of $\rho_0$, which is $\sqrt{a^2+(1-a)^2}$. In the third inequality, we used again the B\"ottcher-Wenzel inequality this time with the $\|G_1\|=1$ and the fact that $G_1$ commutes with $G_9$. 

\end{enumerate}

Now, putting everything together in (\ref{thetwointegrals}), we obtain, 
\begin{widetext}
    $$
\|\hat \rho(\tau)-\hat \rho_0 \|\leq 2\epsilon \sqrt{2} a \|P\|+\int_0^\tau\left(\epsilon \sqrt{2}\|P\|+2 \epsilon^2 \gamma_0 \sqrt{a^2+(1-a)^2}(\tau-r)\right) \|\hat \rho(r)-\rho_0\|dr.
$$
\end{widetext}

We now apply Gr\"onwall's inequality (see, e.g., \cite{JimM}) to obtain 
$$
\|\hat \rho(\tau)-\rho_0\| \leq 2 \epsilon \sqrt{2}a \|P\| e^{\epsilon \sqrt{2}\|P\|\tau+\epsilon^2\gamma_0 \sqrt{a^2+(1-a)^2}\tau^2}. 
$$
If the control steers to the desired final condtion $\rho_1$ in (\ref{realrho}),  then applying the above inequality when $\tau={h_9}$, we obtain 
$$
\| \rho_1-\rho_0\|\leq \frac{2\sqrt{2}}{h_9}a\|P\|e^{\sqrt{2} \|P\|+\gamma_0\sqrt{a^2+(1-a)^2}}.
$$ 
Repeating similar arguments for the case $h_9<0$, we arrive to the general bound for $|h_9|$
\begin{equation}\label{boundh9}
|h_9| \leq \frac{2 \sqrt{2}a\|P\|}{\|\rho_1-\rho_0\|} e^{\sqrt{2}\|P\|+\gamma_0\sqrt{a^2+(1-a)^2}}. 
\end{equation}
Since we have an upper bound for $\|P\|$, this gives a theoretical upper bound for $|h_9|$. Notice that in a numerical search we can use the bound in two ways. We can simply search on a grid for all the allowed values of $\|P\|$ and for the values of $|h_9|$  satisfying the bound (\ref{boundh9})  where $\|P\|$ is replaced by its upper bound. A more efficient approach is to insert the loop for $h_9$ inside the loop for $\|P\|$ and use the {\it current} value of $\|P\|$ in (\ref{boundh9}). 

The following section present a numerical example where we implement the above  theory.

\section{A case study: An Hadamard-like transformation on the density matrix}\label{CS}

 We assume we want to perform a rotation of $\theta=\frac{\pi}{4}$, that is an Hadamard-type  gate,  on the lowest two energy levels  of the initial state $\rho_0$ in (\ref{rho0})  with $a=\frac{2}{3}$, that is, we want to  perform, in time $T=1$,  the state transfer 
 \begin{equation}\label{statetransf}
\rho_0:= \begin{pmatrix}0 & 0 & 0 \cr 0 & \frac{2}{3} & 0 \cr 0 & 0 & \frac{1}{3} \end{pmatrix} \rightarrow \rho_1=\begin{pmatrix} 0 & 0 & 0 \cr 0 & \frac{1}{2} & -\frac{1}{6} \cr 0 & -\frac{1}{6} & \frac{1}{2} \end{pmatrix}. 
 \end{equation}
 Furthermore, we assume in the cost (\ref{occupacost}),(\ref{occupacost2}),  $\gamma_0=1$, that is, we give  double  emphasis on the occupancy as on the energy of the control. We know that the solution pair (trajectory-control) has to satisfy the differential equations (\ref{aux1})-(\ref{aux2}), for appropriate parameter $h_9$ and matrix $P$ of the form (\ref{newmatrixP}). In order to obtain bounds for the norm of the matrix $P$ and the parameter $h_9$, we solve first the problem for the zero occupancy case $\gamma_0=0$. This is done in the following theorem whose proof, given in  Appendix \ref{ZOC1},  is obtained by adapting techniques of \cite{conBenj}, \cite{conben}. 
 
 \begin{theorem}\label{Final}
 A  control $U=U(t)$ achieving  the state transfer (\ref{statetransf}) and minimizing the cost (\ref{occupacost2}) with $\gamma_0=0$, is obtained  by the solution of the system 
 (\ref{aux1})-(\ref{aux2}) with  
 $$
 P=\frac{\sqrt{3} \pi}{2 \sqrt{2}} \begin{pmatrix} 0 & 1 & i \cr -1 & 0 & 0 \cr i & 0 & 0 \end{pmatrix}, 
 $$
 and $h_9=\sqrt{2}\pi$. The minimum cost is $J=\frac{3 \pi^2}{4}$. By defining $A:=ih_9G_9$, The control and trajectory are given by (\ref{explicexpre}).
 \end{theorem} 
 
 In the following we will need to approximate the bounds that we obtain for our parameters as in Section \ref{Red}. In every approximation we will choose a larger number when approximating an upper bound and a smaller number when approximating a lower bound. For example if our numerical calculations give $3.1467$ for an upper bound we may approximate it as $3.15$. Also notice that the optimal in Theorem \ref{Final} and in general for any $\gamma_0$ is far from being unique as highlighted in the subsection on the elimination of redundancy and symmetries  \ref{SRed}. It would be interesting to use this non uniqueness of the control to select the optimal control which has further desirable features, e.g., robustness to parameter and design imperfections. However this is beyond the scope of this paper.  
 
 We now apply the bounds on the control norm $\|P\|$ of Subsection \ref{BNC} to our numerical example ($\gamma_0=1$ and $a=\frac{2}{3}$). In formula (\ref{maximin}) the bound $B_L$ for the cost is given by the cost in the zero occupancy case of Theorem \ref{Final}, i.e., $B_L=\frac{3\pi^2}{4}$, To obtain the upper bound, $B_U$, we apply formula (\ref{inc}) (with $\gamma_2=\gamma_0=1$ and $\gamma_1=0$)  and we calculate the integral 
 \begin{align*}
     \int_0^1\langle G_1, \tilde \rho_1(t) \rangle dt&=\int_0^1\left \langle G_1, e^{At} e^{(-A+P)t} \rho_0 e^{(A-P)t} e^{-At}    \right \rangle dt\\
     &=\int_0^1 \left\langle G_1,  e^{(-A+P)t} \rho_0 e^{(A-P)t}  \right \rangle dt, 
 \end{align*}
 with the matrices $A$ and $P$ described in Theorem \ref{Final}. 
 This gives the upper bound $B_U$ for the cost $B_U=\frac{3\pi^2}{4}+\int_0^1  \left \langle G_1,  e^{(-A+P)t} \rho_0 e^{(A-P)t}  \right \rangle dt$, which gives $B_U\approx {6.4692}$. Plugging this in (\ref{maximin}), we obtain the following bounds for $\|P\|$ 
 \begin{equation}\label{numbound}
 3.4839 \leq \|P\| \leq 3.897. 
 \end{equation}
 As for $h_9$, we use (\ref{boundh9}) with $\gamma_0=1$, $a=\frac{2}{3}$ and notice that the right hand side is an increasing function of $\|P\|$ so that replacing for $\|P\|$ the upper bound in (\ref{numbound}) we obtain an upper bound for $|h_9|$. We also find it convenient to rescale $h_9$ with $\tilde h_9:=\frac{h_9}{\sqrt{2}}$ so that the reference value which is the value of $\tilde h_9$ for the zero occupancy case is simply $\tilde h_9=\pi$. With this, we  calculated the following bound 
 \begin{equation}\label{tilh9bound}
 |\tilde h_9| \leq 8128. 
 \end{equation}

 \vspace{0.5cm}
 
 Our numerical search was organized as follows. We first run the MATLAB simulation of system (\ref{aux1}),  (\ref{aux2}) for values of $|\tilde h_9|$ in the range $[0,80]$ 
  for all the values of $\|P\|$ in the range described in (\ref{numbound}), by using a stepsize of $\Delta_h=0.03$ for the values of $\tilde h_9$ and a stepsize approximately (upper bound minus lower bound in (\ref{numbound}) divided by $100$) $\Delta_P\approx0.004131$ for $\|P \|$. We recorded  in a vector the minimum distance square from the desired value $\rho_1$ in (\ref{statetransf}) as a function of $|\tilde h_9|$. The plot in Figure \ref{Figura1} shows the behavior such a minimum  distance square as a function $|\tilde h_9|$. Such a distance starts converging to the initial (worst case) value of $\|\rho_0-\rho_1\|^2=\frac{1}{9}$ much earlier than the theoretical bound in (\ref{tilh9bound}). We run a separate simulation for values of $|\tilde h_9|$ in the interval $[80, 8128]$ (with larger stepsizes) which confirmed  this trend and that the error square remains away from zero for larger values of $|\tilde h_9|$, and tends to the initial value of  $\frac{1}{9}$.

\begin{figure}
\centering
\includegraphics[scale=0.6]{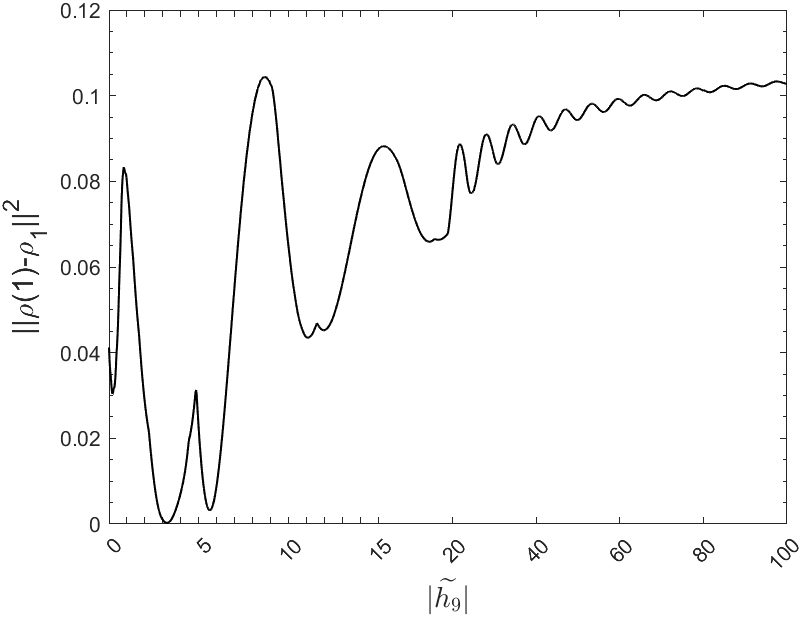}
\caption{Minimum distance square ($\|\rho(1)-\rho_1\|^2$) as a function of $|\tilde h_9|$ for $|\tilde h_9|\in [0,80]$.  }
\label{Figura1}
\end{figure}

From Figure \ref{Figura1} it appears that there are two minima for the distance square one of which is near the value of $\pi$  which was the optimal value for the zero occupancy case. 
 From simulations with smaller stepsizes it appeared  that the second minimum is always above the first one. Therefore, we focused our search on an interval around the first minimum which we took as $|h_9 | \in [3.18, 3.3]$. For the values of $| \tilde h_9|$ in this interval we run a similar simulation code obtaining the minimum value of the square error as a function 
 of $\|P\|$ with stepsizes $\Delta_h=0.01$ for $|\tilde h_9|$, $\Delta_P=0.01$ and $\Delta_{\theta_1}=\Delta_{\theta_2}=\frac{\pi}{100}$ for the angles that appear in the specification of the initial control parameters.
  We have parametrized the initial matrices $P$ for a fixed $\|P\|$ with two angles $\theta_1$ and $\theta_2$,  and we choose a finer stepsize of $\Delta_{\theta_1}=\Delta_{\theta_2}$  for these angles as well   We obtained the plot in Figure \ref{Figura2} describing the dependence of the minimum error square on $\|P\|$. From this plot it appears that $\|P\|$ has to be in an interval smaller than the original one in (\ref{numbound})  fr which we took  equal to $[3.7839,3.8539]$.

 \begin{figure}
\centering
\includegraphics[scale=0.6]{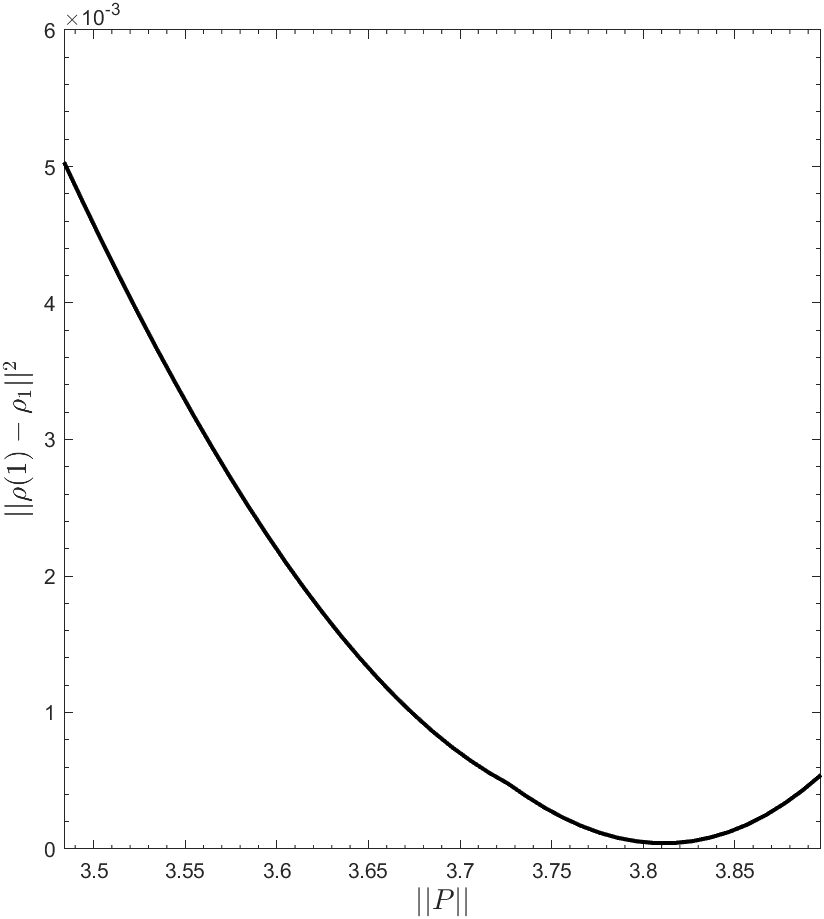}
\caption{Minimum distance square ($\|\rho(1)-\rho_1\|^2$) as a function of $\|P\|$ for $|\tilde h_9|\in [3.18,3.3]$.  }
\label{Figura2}
\end{figure}

We run a third  program for values 
of $|\tilde h_9|$ and $\|P\|$ in $[3.18,3.3]\times [3.7839,3.8539]$ , with  stepsizes $0.01$ for both $\tilde h_9$ and $\|P\|$. Within a single value of $\|P\|$ for the angles $\theta_1$ and $\theta_2$ used in the specification of $P$ the stepsize was taken $\frac{\pi}{100}$. This time we recorded  every value of the parameters which gave a final  square error below 
a threshold which we fixed at $5\times 10^{-5}$. We also recorded the square error  along with the corresponding value of the cost. This was calculated from the formula (cf. Subsection \ref{BNC} with $\gamma_0=1$)  $J=\frac{1}{2}\|P\|^2+ 2\int_0^1 \langle G_1, \rho(t)\rangle dt$). We picked  the values of the parameters which gave the minimum cost. These were: $u_4(0)=0$, $u_5(0)\approx 2.6877$, $u_6(0)approx 2.6891$, $u_7\approx 0.0845$, $\tilde h_9\approx 3.253$, with a cost $J\approx=7.597$.  This is, within unavoidable numerical error and with good approximation, the optimal control.    

The  plots of(\ref{Figura3}) (\ref{Figura4}) give the the time dependence of the optimal control and trajectory respectively.

  \begin{figure}
\centering
\includegraphics[scale=0.85]{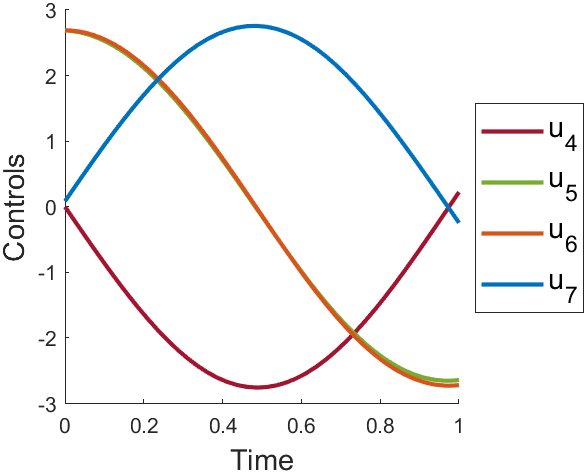}
\caption{Time dependence for the optimal controls for the problem with $\gamma_0=1$ and state transfer (\ref{statetransf}) }
\label{Figura3}
\end{figure}

  \begin{figure}
\centering
\includegraphics[scale=0.85]{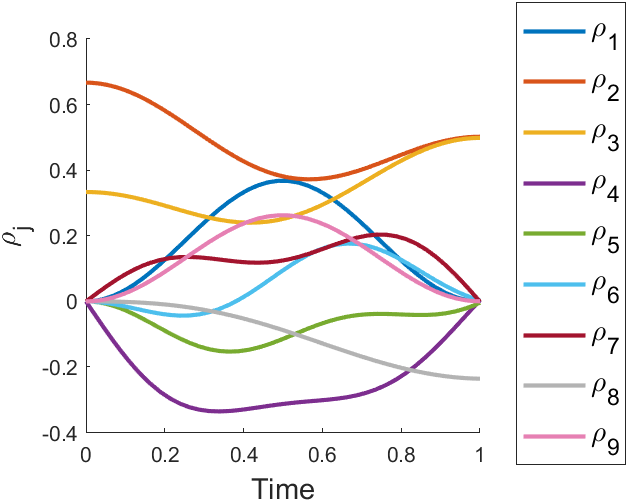}
\caption{Time dependence for the optimal density matrix for the problem with $\gamma_0=1$ and state transfer (\ref{statetransf}) }
\label{Figura4}
\end{figure}

\section{Conclusions} 

In this work, we developed a framework for optimal state transfer between mixed states for the quantum Lambda system, a configuration appearing in many quantum information processing architectures. Our approach  moves beyond adiabatic protocols to design arbitrarily fast, non-adiabatic controls. The optimal control we found   balances high fidelity and the energy of the control field. 

In order to produce a feasible solution of the optimal control design problem, we applied several techniques to bound the number and size of parameters in the resulting finite dimensional optimization search. These techniques include symmetry reduction, scaling, and averaging methods. In this respect, we believe that the system   treated here can be used as a reference to apply similar techniques to different quantum optimal control problems. For such problems the mere application of the Pontryagin’s Maximum Principle \cite{BoscaRev}, is usually not sufficient to compute the optimal control, and auxiliary techniques, such as the ones developed in this paper, are needed.

We illustrated our theory  with a numerical case study for which we explicitly  found  optimal controls and trajectories that implement a Hadamard-like gate on a mixed state,  while balancing the trade-off between control field energy  and occupancy in the unstable, higher energy level. This work provides possible pathways to developing faster, more accurate quantum operations to help build scalable, reliable quantum information processing platforms.

\begin{acknowledgments}
This material is based upon work supported   in part by the U. S. Army Research Laboratory and the U. S. Army Research Office under contract/grant number W911NF2310255. 
 D. D'Alessandro also would like to acknowledge support from a Scott Hanna Professorship at Iowa State University. 
 
\end{acknowledgments}

\appendix
\section{Proof of the bound (\ref{UB3})}\label{upboun}

After a standard change of variables, writing $\beta:=b+id$, 
\begin{widetext}
    $$
\int_0^\tau [e^{iG_9 s} P e^{-iG_9s}, \rho_0] ds=
$$
$$
\sqrt{2}\int_0^{\frac{\tau}{\sqrt{2}}} \begin{pmatrix} 0 &a(p\cos(x)-\beta \sin(x)) & (1-a)(p \sin(x)+\beta \cos(x)) \cr 
a(p\cos(x)-\beta^* \sin(x)) & 0 & 0 \cr 
(1-a)(p \sin(x)+\beta^*\cos(x))& 0 & 0  \end{pmatrix} dx=
$$
$$
\sqrt{2} \begin{pmatrix} 0 & z & w \cr z^* & 0 & 0 \cr w^* & 0 & 0 \end{pmatrix}, 
$$
\end{widetext}

with $z:=a\left (p \sin(\frac{\tau}{\sqrt{2}})-\beta(1-\cos(\frac{\tau}{\sqrt{2}}))\right),$ 
$w=(1-a)\left(p ( 1-\cos(\frac{\tau}{\sqrt{2}}))+ \beta \sin(\frac{\tau}{\sqrt{2}})  \right).  
$
Calculate the square norm of the above matrix which is (use the shorthand notation $s$ for $\sin(\frac{\tau}{\sqrt{2}})$ and $c$ for  $\cos(\frac{\tau}{\sqrt{2}}$)).\footnote{Also recall that we have chosen $a>1-a$. We use Young's inequality $2xy\leq x^2+y^2$. } 
\begin{widetext}
    $$
4(|z|^2+|w|^2)=
$$
$$
4\left( a^2(p^2s^2+|\beta|^2(1-c)^2 ) +(1-a)^2(p^2(1-c)^2+|\beta|^2s^2)+
\left( (1-a)^2-a^2  \right) (2Re(\beta)ps (1-c)) \right) \leq 
$$ 
$$
4\left( a^2(p^2s^2+|\beta|^2(1-c)^2 ) +(1-a)^2(p^2(1-c)^2+|\beta|^2s^2)+
\left( a^2 -(1-a)^2 \right) (p^2(1-c)^2+|\beta|^2s^2) \right)
$$
$$
=4 a^2(p^2+|\beta|^2)(s^2+(1-c)^2) =8 a^2(p^2+|\beta|^2)(1-c)=4a^2\|P\|^2(1-c)\leq 8a^2\|P\|^2. 
$$
\end{widetext}

Taking the square root, we obtain the bound (\ref{UB3}). ,

\section{The zero occupancy case for the case study of Section \ref{CS}; Proof of Theorem \ref{Final}}\label{ZOC1}

 Define $A:=i h_9 G_9:=\begin{pmatrix} 0 & 0 & 0 \cr 0 & 0 & a \cr 0 & -a & 0 \end{pmatrix}$ so as to write the solution of (\ref{aux1}) (\ref{aux2}) as in (\ref{explicexpre}) 
 The matrices  $A$ and $P$ are chosen so that 
 \begin{equation}\label{conditions}
  e^{A}e^{-A+P}\rho_0e^{A-P} e^{-A}=\rho_1, 
 \end{equation}
 and 
 \begin{align}\label{AandPbis}
 A&:=\begin{pmatrix} 0 & 0 & 0 \cr 0 & 0 & a \cr 0 & -a & 0  \end{pmatrix},\\
 P&:=r \begin{pmatrix}0 & \cos(\theta) & \sin(\theta) e^{i\psi} \cr -\cos(\theta) & 0 & 0 \cr -\sin(\theta) e^{-i\psi} & 0 & 0  \end{pmatrix}, \notag
 \end{align}
 for $r >0$. 
 The cost $J$, which is equal to $\frac{1}{2}\|P\|^2=r^2$,  must be minimized. Therefore the problem becomes the following:
 
 \vspace{0.5cm}
 
 \noindent{\bf Problem:} {\it Find $A$ and $P$ in (\ref{AandPbis}) such that (\ref{conditions}) are verified with a minimum value for $r>0$}. 
 
 \vspace{0.25cm}

 \noindent It is easily verified that if $X_f\rho_0X_f^\dagger=\rho_1$ with $\rho_0$ and $\rho_1$ of the form $\begin{pmatrix} 0 & 0 \cr 0 & R  \end{pmatrix}$ with $R$ a $2 \times 2$ matrix 
 with full rank, as in our case, $X_f$ must be block diagonal. Since  in our case $X_f:=e^Ae^{-A+P}$ and $e^{A}$  is block diagonal,  this  implies that $e^{-A+P}$ is block diagonal. This implies the following 
 
 \begin{lemma}\label{scalar} (cf. Proposition II.3 in \cite{conBenj}) Given (\ref{AandPbis}), define $T\in SU(3)$ as 
 $$
 T:=\begin{pmatrix} 1 & 0 & 0 \cr 0 & \cos(\theta) & \sin(\theta) e^{i\psi} \cr 0 & -\sin(\theta) e^{-i\psi} & \cos(\theta) \end{pmatrix}. 
 $$
 Then 
 \begin{equation}\label{mild}
 Te^{-A+P} T^\dagger=\begin{pmatrix} e^{i\phi} & 0 & 0 \cr 0 & e^{i\phi} & 0 \cr 0 & 0 & e^{-2i\phi} \end{pmatrix}, 
 \end{equation}
 for some $\phi \in \mathbb{R}$. If in addition $|\cos(\theta)|\not=|\sin(\theta)|$ or $\cos(\psi)\not=0$, then $ Te^{-A+P} T^\dagger$, and therefore $ e^{-A+P} $,  is a scalar matrix. 
 \end{lemma}
 \begin{proof}
 Let $X_f=e^{-A+P}$. Since $X_f$ commutes with $-A+P$ we get $-X_fAX_f^{\dagger} + X_f P X_f^\dagger=-A+P$. Separating the block diagonal and the block antidiagonal parts in this equality we get that $X_f$ commutes with {\it both} $A$ and $P$. Equivalently $TX_fT^\dagger$ commutes with both $TAT^\dagger$ and $TPT^\dagger$. Let us calculate $TPT^\dagger$ with $P$ in (\ref{AandPbis}) which gives 
 $$
 TPT^\dagger:=r\begin{pmatrix} 0 & 1 & 0 \cr -1 & 0 &0 \cr 0 & 0 & 0 \end{pmatrix}.
 $$
 By writing $TX_fT^\dagger$ as a general block diagonal matrix in $SU(3)$ and imposing that it commutes with $TPT^\dagger$ we obtain the form (\ref{mild}).

 Use now the fact that $TX_fT^\dagger$ commutes with $TAT^\dagger$ where $TAT^\dagger$ is calculated as 
 $$
 a\begin{pmatrix}  0 & 0 & 0 \cr 0 & -i\sin(2\theta) \sin(\psi) & \cos^2(\theta)+\sin^2(\theta) e^{2i\psi} \cr 0 & -(\cos^2(\theta)+\sin^2(\theta)e^{-2i\psi}) & i \sin(2\theta) \sin(\psi) . \end{pmatrix}. 
 $$
 Commutativity of $TX_fT^\dagger$ in (\ref{mild}) with $TAT^\dagger$ is equivalent to commutativity of $TX_fT^\dagger$ with $\begin{pmatrix} 0 & 0 & 0 \cr 0 & 0 & w \cr 0 & -w^* & 0\end{pmatrix}$, with $w:=\cos^2(\theta)+\sin^2(\theta) e^{2i\psi}$ and if $w \not=0$ implies that $TX_fT^\dagger$ and therefore $X_f$ is a scalar matrix. $w=0$ is equivalent to 
 $e^{-i\psi}w=\cos^2(\theta) e^{-i\psi}+\sin^2(\theta)e^{i\psi}=\cos(\psi)+i\sin(\psi)(\sin^2(\theta)-\cos^2(\theta))=0$. Therefore if   $\cos(\psi)\not=0$ or $|\cos(\theta)|\not=|\sin(\theta)|$,  $X_f$ is a scalar matrix as claimed. 
  \end{proof}

According to the previous Lemma, we have to consider two cases: {\bf Case 1}:  $\cos(\psi) \not=0$ and/or $\sin^2(\theta) \not= \cos^2(\theta)$, in which case $e^{-A+P}$ must be a scalar matrix (in $SU(3)$);  {\bf Case 2}:  $\cos(\psi) =0$ and  $\sin^2(\theta) = \cos^2(\theta)$, in which case $Te^{-A+P}T^\dagger=\texttt{diag}(e^{i\phi}, e^{i\phi}, e^{-2i\phi})$. In the following two subsections we consider the two cases. We remark that although we consider the state transformation (\ref{statetransf}) the treatment can be extended to general state transformations.

\subsection{Case 1: $e^{-A+P}$ scalar}\label{scalar1} 
 Since $e^{-A+P}$ is scalar $\rho_1=e^{A}e^{-A+P}\rho_0 e^{A-P}e^{-A}=e^{A}\rho_0 e^{-A}$ which is equivalent to $\sin(2a)=1$, that is, 
 \begin{equation}\label{a}
a=\frac{\pi}{4}+l\pi, \qquad l \, \texttt{integer}. 
 \end{equation}
 $e^{-A+P}$ is a scalar matrix if and only if 
 \begin{equation}\label{IAP}
 i(-A+P)=\begin{pmatrix} 0 & ir \cos(\theta) & ir \sin(\theta) e^{i\psi} \cr -ir \cos(\theta) & 0 & - ia \cr -ir\sin(\theta) e^{-i\psi} & i a & 0  \end{pmatrix}, 
 \end{equation}
 has three coinciding (real) eigenvalues up to multiples of $2\pi$, that is,  
 \begin{align}\label{lambda123B}
 \lambda_1&=\frac{2k\pi}{3}, \quad \lambda_2=\frac{2k\pi}{3}+2m\pi,\\ \lambda_3&=-(\lambda_1+\lambda_2)=\frac{2k\pi}{3}-2(m+k)\pi.  \notag
 \end{align}
 Calculating the characteristic polynomial of the matrix $i(-A+P)$ in (\ref{IAP}) and imposing this condition leads to 
 \begin{widetext}
 $$
 \lambda^3-(r^2+a^2) \lambda+\sin(2\theta)\sin(\psi)ar^2=\lambda^3-\left[\lambda^2_1+\lambda_2^2+\lambda_1 \lambda_2\right]\lambda+\lambda_1 \lambda_2(\lambda_1+\lambda_2), 
 $$
 \end{widetext}
 
 which leads to the two conditions
 $$
 r^2+a^2=\lambda^2_1+\lambda_2^2+\lambda_1 \lambda_2, 
 $$
 $$
 \lambda_1 \lambda_2(\lambda_1+\lambda_2)=\sin(2\theta) \sin(\psi) a r^2. 
 $$
 Since $\sin(2\theta)$ and $\sin(\psi)$ are free parameters subject to the condition that are not both with absolute value equal to $1$ (otherwise we would be in the situation of {\bf Case 2}), defining $\phi_l=|a|$,\footnote{Note the dependence on $l$ is highlighted here because of (\ref{a}).}  the problem becomes to minimize 
 \begin{equation}\label{erre2}
 r^2=\lambda_1^2+\lambda_2^2+\lambda_1 \lambda_2-\phi_l^2>0, 
 \end{equation} 
 subject to $-\phi_lr^2< \lambda_1\lambda_2(\lambda_1+\lambda_2)<\phi_l r^2$, which is , 
 \begin{widetext}
     \begin{equation}\label{con3}
 -\phi_l\left( \lambda_1^2+\lambda_2^2+\lambda_1 \lambda_2-\phi_l^2\right)< \lambda_1 \lambda_2(\lambda_1+\lambda_2) <\phi_l\left( \lambda_1^2+\lambda_2^2+\lambda_1 \lambda_2-\phi_l^2\right), 
 \end{equation}
 \end{widetext}
 
 where $\lambda_1,$ $\lambda_2$ and $\phi_l$ are restricted to values given in (\ref{lambda123B}) and (\ref{a}). Therefore, the problem is a {\it constrained integer optimization problem}. To solve it, it is first useful to eliminate a factor $2\pi$ which appears everywhere and to define $\hat r,\hat\phi_l,\hat \lambda_1,\hat  \lambda_2:=\frac{\{ r,\phi_l, \lambda_1,\lambda_2\} }{2\pi}$.  
  Therefore the problem becomes to minimize the hatted quantity corresponding to (\ref{erre2}) with the hatted constrained corresponding to (\ref{con3}), with (from (\ref{lambda123B}))
  \begin{equation}\label{hatted} 
 \hat   \lambda_1=\frac{k}{3}, \quad \hat \lambda_2=\frac{k}{3}+m, \quad \hat \phi_l:=\left |\frac{1}{8}+\frac{l}{2}\right|.   
  \end{equation}
 A direct verification shows that condition (\ref{con3}) splits into two conditions of the form 
 \begin{equation}\label{ineq1}
 \left( \hat \phi_l +\hat \lambda_1\right)  \left( \hat \phi_l -(\hat \lambda_1+\hat \lambda_2)\right)  \left( \hat \phi_l +\hat \lambda_2\right)<0
 \end{equation}
  \begin{equation}\label{ineq2}
 \left( \hat \phi_l -\hat \lambda_1\right)  \left( \hat \phi_l +(\hat \lambda_1+\hat \lambda_2)\right)  \left( \hat \phi_l -\hat \lambda_2\right)<0. 
 \end{equation}

 \vspace{0.25cm}

 \noindent {\bf Analysis of the region described by (\ref{ineq1}) (\ref{ineq2})} 
 
 \vspace{0.25cm}
 
 Since we can arbitrarily permute $\lambda_1$, $\lambda_2$ and $\lambda_3=-(\lambda_1+\lambda_2)$ we shall assume, without loss of generality, that 
 \begin{equation}\label{order}
 \hat \lambda_1\geq \hat \lambda_2 \geq -(\hat \lambda_1+\hat \lambda_2).
 \end{equation}
 Notice that this implies $\hat \lambda_1 \geq 0$. 
 Under this assumption, let us analyze, which sign combinations are possible in (\ref{ineq1})  and (\ref{ineq2}). Let us start with (\ref{ineq1}).

 The sign combination $++-$ is not possible because we would have $\hat \phi_l > \hat \lambda_1 + \hat \lambda_2$ and $\hat \phi_l < -\hat  \lambda_2$ which would imply $\hat \lambda_2 < -(\hat \lambda_1+\hat \lambda_2)$ which contradicts (\ref{order}). Similarly the combination $-++$ is not possible because $\hat \phi_l < -\hat \lambda_1$ and $\hat \phi_l > \hat \lambda_1+ \hat \lambda_2$ would imply $\hat \lambda_1<-(\hat \lambda_1+\hat \lambda_2)$ which again contradicts (\ref{order}).  In fact any combination $-**$ is not possible, since $\hat \phi_l < -\hat \lambda_1$ would imply $\hat \phi_l <0$. Therefore we only consider the sign combination $+-+$ for (\ref{ineq1}) 
 
 Let us now analyze (\ref{ineq2}). The combination $+-+$ is not possible because (again together with (\ref{order})) would mean $\hat \phi_l > \hat \lambda_1 \geq -(\hat \lambda_1+\hat \lambda_2) > \hat \phi_l$, a contradiction. Analogously, the combination $+ + -$ is not possible because it would imply $\hat \phi_l > \hat \lambda_1 \geq \hat \lambda_2 > \hat \phi_l$, another contradiction. This leaves the combinations $---$ and $-++$ for (\ref{ineq2}).

 Now the combination $+-+$ for (\ref{ineq1}) together with $---$ for (\ref{ineq2}) is not possible because $\hat \phi_l > -\hat \lambda_2$ and $\hat \phi_l < \hat \lambda_2$  would imply $\hat \lambda_2>0$. However this would give that $\hat \phi_l < -(\hat \lambda_1+\hat \lambda_2) <0$ which has to be excluded. We are left with the combination $+-+$ for (\ref{ineq1}) and $-++$ for (\ref{ineq2}). This gives that the following inequalities have  to be satisfied simultaneously. From (\ref{ineq1}) 
 
 \begin{equation}\label{Fineq1}
  \hat \phi_l > -\hat \lambda_1, \qquad  \hat \phi_l < \hat \lambda_1+ \hat \lambda_2, \qquad \hat \phi_l > -\hat \lambda_2;
  \end{equation}
 From (\ref{ineq2}) 
  \begin{equation}\label{Fineq2}
  \hat \phi_l <  \hat \lambda_1, \qquad  \hat \phi_l > - \hat \lambda_1- \hat \lambda_2, \qquad \hat \phi_l > \hat \lambda_2;
  \end{equation}
 From the first one of (\ref{Fineq2}) we get that $\hat \lambda_1$ is strictly positive which makes the first on of (\ref{Fineq1}) obvious. (\ref{order}) implies that the third one of (\ref{Fineq2}) implies the second one of (\ref{Fineq2}) which is therefore  redundant. Combining the remaining inequality, the region to explore for $\hat \phi_l$ is 
 \begin{equation}\label{region}
 0\leq |\hat \lambda_2 | < \hat \phi_l < \min \{ \hat \lambda_1, \hat \lambda_1+\hat \lambda_2 \} 
 \end{equation}

 We now consider the definitions of $\hat \lambda_1$ and $\hat \lambda_2$ in terms of $k$ and $m$ and identify the values of $k$ and $m$ which are compatible with (\ref{region}). We need both $\hat \lambda_1 > |\hat \lambda_2|$ and $\hat \lambda_1+\hat \lambda_2 > |\hat \lambda_2|$. The first inequality gives $-2 \frac{k}{3} < m< 0$ which implies $\hat \lambda_1+\hat \lambda_2=\frac{2k}{3}+m>0$. With this, the second inequality, which is $\left| \frac{k}{3}+m\right| < \frac{2k}{3}+m$,  gives $-\frac{k}{2}<m$. Therefore the values of $k >0$ and $m$ to be considered are such that $m$ is in the interval $\left( -\frac{k}{2}, 0 \right)$. In particular 
 
 \begin{enumerate}
 
 \item (Case A) If $m \in \left[ -\frac{k}{3}, 0 \right)$ the minimum in (\ref{region}) is $\frac{k}{3}$. 
 
 \item (Case B) If $m \in \left( \frac{
 -k}{2}, \frac{-k}{3} \right)$,  the minimum in (\ref{region}) is $\frac{2k}{3}+m$. 
 
 \end{enumerate}
 
 \vspace{0.25cm}

  \noindent {\bf Minimization of $\hat r^2$} 
  
   \vspace{0.25cm}

  The minimization of $r^2$ in (\ref{erre2}) is obtained with a min/min procedure. One first fix $k$ and $m$ in (\ref{hatted}) so that $m$ is 
  in the interval $\left( -\frac{k}{2}, 0 \right)$ as described above   and  finds $l$ so that $\hat \phi_l^2$ is maximized over the region (\ref{region}). 
 This gives for $\hat r^2$ a function of $k$ and $m$ only which is then minimized. It is  convenient to write $k:=3j+s$ with $s=0,1,2$. Notice that both $j$ and $s$ cannot be zero since $k>0$. Let us consider the above cases $A$ and $B$ separately. \\
 
 \noindent{}Case A:=
 \begin{enumerate}
 \item If $s=0$, the maximum of $\hat \phi_l$ is $\hat \phi_l:=\frac{k}{3}-\frac{1}{8}$. It is achieved for $l=-\frac{2 k}{3}=-2j$ 
  \begin{proof}
  If $l\leq 0$ the constraint is written as $\frac{l}{2}+\frac{1}{8} <j$ which gives $l< 2j -\frac{1}{4}$ with the maximum achieved at $l=2j-1$. Thus the value of $\hat \phi_l$ is 
  $\hat \phi_l =\frac{l}{2}+\frac{1}{8}=\frac{2j-1}{2}+\frac{1}{8}=j-\frac{3}{8}$. If $l<0$, $\hat \phi_l=\frac{|l|}{2}-\frac{1}{8}.$ It is 
  maximized for $|l|=2j$ which gives the value for
   $\hat \phi_l=j-\frac{1}{8}
   =\frac{k}{3}-\frac{1}{8}$. 
   Since this value is bigger then the previous one, this is the maximum. 
  \end{proof}
The proofs of the following statements are similar.
  \item If $s=1$, the maximum of $\hat \phi_l$ is $\hat \phi_l:=\frac{k-1}{3}+\frac{1}{8}$. It is achieved for $l=\frac{2 (k-1)}{3}=2j$ 
 
 \item If $s=2$, the maximum of $\hat \phi_l$ is $\hat \phi_l:=\frac{k-2}{3}+\frac{5}{8}$. It is achieved for $l=\frac{2 (k-2)}{3}+1=2j+1$ 
 
 \end{enumerate}

 \noindent{}Case B:=
 \begin{enumerate}
 
 \item If $s=0$, the maximum of $\hat \phi_l$ is $\hat \phi_l:=2\frac{k}{3}+m-\frac{1}{8}$. It is achieved for $l=-4j-2m= -\frac{4 k}{3}-2m$  (Notice the continuity with Case A when $m=\frac{-k}{3}$). 
 
   \item If $s=1$, the maximum of $\hat \phi_l$ is $\hat \phi_l:=2 \frac{k-1}{3}+m+\frac{5}{8}$. It is achieved for $l=\frac{4 (k-1)}{3}+2m+1=4j+2m+1$. 
   
   \item  If $s=2$, the maximum of $\hat \phi_l$ is $\hat \phi_l:=2 \frac{k-2}{3}+m+\frac{9}{8}$. It is achieved for $l=\frac{4 (k-2)}{3}+2m+2=4j+2m+2$.

 \end{enumerate}
 
 \vspace{0.5cm}

 Let us now write the `hatted' version of $r^2$ in (\ref{erre2}) by replacing the values $\hat \lambda_1$ and $\hat \lambda_2$ in (\ref{hatted}). Since we want to 
 minimize this over $k$ and $m$, for $m$ in the interval $\left( - \frac{k}{2}, 0 \right)$ we replace $\hat \phi_l^2$ with the {\it maximum}  over $l$ of $\hat \phi_l$ which is calculated as above and it is a function of $k$ and $m$. We denote this function  by $F=F(k,m)$. Therefore the problem is to minimize over $k$ and $m$, 
 \begin{align}\label{errehat}
 \hat r^2&=\hat \lambda_1^2+\hat \lambda_2^2+\hat \lambda_1 \hat \lambda_2-F(k,m)\\
 &= \frac{k^2}{3}+mk+m^2-F(k,m). \notag
 \end{align} 
 We fix $k>0$ and minimize such a function over $m \in \left( -\frac{k}{2}, 0\right)$, and then minimize  over $k$. If $m \in \left[ -\frac{k}{3},0 \right)$,  $F(k,m)$ does not depend on $m$. The only dependence on $m$ in (\ref{errehat})  comes from the term $mk+m^2$ which is increasing with $m$, since we have $m>\frac{-k}{2}$, the minimum is obtained at the smallest  integer  
 greater than or equal to $-\frac{k}{3}$. As before we write $k=3j+s$, $s=0,1,2$ and consider the three cases, using the expressions of $F(k)$ calculated above in the Case A. 
 In all cases the minimizing $m$ is $m=-j$.  
 
 \begin{enumerate}
 
 \item If $s=0$ the minimum $\hat r^2$ as a function of $j$ is 
 \begin{equation}\label{CA1}
 \hat r^2=\frac{k^2}{3}-jk+j^2-\left( j -\frac{1}{8} \right)^2=\frac{j}{4}-\frac{1}{64}. 
 \end{equation}

 \item If $s=1$ the minimum $\hat r^2$ as a function of $j$ is 
 \begin{equation}\label{CA2}
 \hat r^2=\frac{k^2}{3}-jk+j^2-\left( j +\frac{1}{8} \right)^2=\frac{3}{4}j+\frac{61}{192}
 \end{equation}
 
 \item If $s=2$ the minimum $\hat r^2$ as a function of $j$ is 
 \begin{equation}\label{CA3}
 \hat r^2=\frac{k^2}{3}-jk+j^2-\left( j +\frac{5}{8} \right)^2=\frac{3}{4}j+\frac{181}{192}
 \end{equation}
 
 \end{enumerate}

 We now consider the Case B above, that is, $m \in \left( -\frac{k}{2}, -\frac{k}{3} \right)$. An examination of all cases ($s=0,1,2$) above shows that the minimum $\hat \phi_l$ is 
 always of the form $\hat \phi_l=m+f_s(k)$, where in the cases $s=0,1,2$, we have $f_0(k)=\frac{2k}{3}-\frac{1}{8}$, $f_1(k)=\frac{2k}{3}-\frac{1}{24}$, $f_2(k)=\frac{2k}{3}-\frac{5}{24}$, respectively. With these expressions $\hat r$ has the form (after the maximization of $\hat \phi_l$) 
 \begin{align}\label{r45}
 \hat r^2&=\frac{k^2}{3}+km+m^2-\left( m+f_s(k) \right)^2\\
 &= \frac{k^2}{3}+\left(k-2f_s(k)\right)m-f_s^2(k), \notag
 \end{align}
 and a direct verification of the three case, $s=0,1,2$,  shows that the coefficient $\left(k-2f_s(k)\right)$ is always $<0$, so that the function decreases with $m$. The minimum (for fixed $k$) is achieved at the largest integer (if any) which is strictly less than $-\frac{k}{3}$ and strictly greater than $-\frac{k}{2}$. By writing $k=3j+s$,  such an integer $m$ is in all cases $m=-j-1$ with the restriction that $j>2-s$. Replacing in (\ref{r45}) we find the three cases similarly to (\ref{CA1}) (\ref{CA2}) (\ref{CA3})
 
  \begin{enumerate}
 
 \item If $s=0$ the minimum $\hat r^2$ as a function of $j$ is (replacing $f_0(k)$ in (\ref{r45})) and for $m=-j-1$, 
 \begin{align}\label{BA1}
 \hat r^2&=\frac{(3j)^2}{3}+ \left( 3j-2(j-\frac{1}{8}) \right)m - \left( j-\frac{1}{8} \right)^2\\
 &=j^2-j-\frac{17}{64} \qquad (j>2) \notag
 \end{align}

 \item If $s=1$ the minimum $\hat r^2$ as a function of $j$ is (replacing $f_1(k)$ in (\ref{r45})) and for $m=-j-1$, 
 \begin{align}\label{BA2}
 \hat r^2&=\frac{(3j+1)^2}{3}+\left( 3j+1-2(2j+\frac{5}{8})\right)m-\left( 2j+\frac{5}{8}\right)^2\\
 &=\frac{3j}{4}+\frac{37}{192}. \qquad( j>1) \notag
 \end{align}
 
 \item If $s=2$ the minimum $\hat r^2$ as a function of $j$ is (replacing $f_2(k)$ in (\ref{r45})) and for $m=-j-1$, 
 \begin{align}\label{BA3}
 \hat r^2&=\frac{(3j+2)^2}{3}+\left( 3j+2-2(2j+\frac{27}{24}) \right)m - \left( 2j+\frac{27}{24}\right)^2\\
 &=\frac{3}{4}j+\frac{61}{192} \qquad ( j > 0) \notag
 \end{align}
 
 \end{enumerate}

 Finally we have to minimize over $k$ the previous expressions in (\ref{CA1})-(\ref{CA3}, (\ref{BA1})-(\ref{BA3}). Consider again $k=3j+s$ and the three possibilities $s=0,1,2$. By comparing (\ref{CA1}) (which has the only restriction $j \geq 1$) and (\ref{BA1}), we find that the minimum is achieved for $j=1$ and it is $\frac{15}{64}$. For $s=1$ by comparing (\ref{CA2}) (which has the only restriction $j \geq 1$) and (\ref{BA2}), we find that the minimum is achieved for $j=1$ and it is $\frac{3}{4}+\frac{61}{192}$. Finally for $s=2$ by comparing (\ref{CA3}) (which has the only restriction $j \geq 1$) and (\ref{BA3}), we find that the minimum is achieved for $j=1$ and it is $\frac{3}{4}+\frac{61}{192}$. Comparing the three cases, the minimum is $\frac{15}{64}$ achieved for $j=1$ $s=0$, which is $k=3$, $m=\frac{-k}{3}=-1$. The corresponding value of $l$ is obtained from the first subcase of case A above and it is $l=-2$. By multiplying $\hat r$ by $2\pi$ we obtain the possible minimum value for $r^2$ with the constraint that $|\sin(2\theta)\sin(\psi)|<1$ in (\ref{AandPbis}). We summarize:
 \begin{lemma}\label{summary2}
 The minimum possible value for the cost  $J=r^2$ under the assumption that $|\sin(2\theta)\sin(\psi)|<1$ in (\ref{AandPbis}) is 
 $$
 r^2_{min}=\frac{15 \pi^2}{16}. 
 $$
 \end{lemma}

 \subsection{Case 2: $e^{-A+P}$ possibly nonscalar}\label{possnsc}
 
 In this case $|\sin(2\theta)|=1$ and $|\sin(\psi)|=1$. Let us rewrite $A$ and $P$ in (\ref{AandPbis})
  as (notice we take   $r>0$ using the observations in Subsection \ref{SRed})
 \begin{equation*}
      P=\begin{pmatrix} 0 & \frac{r}{\sqrt{2}} & \frac{ir}{\sqrt{2}} \cr -\frac{r}{\sqrt{2}} & 0 & 0 \cr  \frac{ir}{\sqrt{2}} & 0 & 0 \end{pmatrix}, \qquad A=\begin{pmatrix} 0 & 0 & 0 \cr 0 & 0 & a \cr 0 & -a & 0 \end{pmatrix}.
 \end{equation*}

 Consider the similarity transformation $(-A+P) \rightarrow T (-A+P) T^\dagger$ with 
 $$
 T:=\begin{pmatrix} 1 & 0 & 0 
 \cr 0 & \frac{1}{\sqrt{2}} & \frac{i}{\sqrt{2}}
  \cr 0 &   \frac{i}{\sqrt{2}} & \frac{1}{\sqrt{2}} 
 \end{pmatrix}, 
 $$
 which gives
 $$
 T (-A+P) T^\dagger=\begin{pmatrix} 0 & r & 0 \cr -r &  ia & 0 \cr 0 & 0 &  ia \end{pmatrix}. 
 $$
 We explicitly calculate $e^{T(-A+P)T^\dagger}$ which gives 
 \begin{widetext}
 $$
 e^{T(-A+P)T^\dagger}=\begin{pmatrix} e^{ i \frac{a}{2}} & 0 & 0 \cr 0 & e^{ i \frac{a}{2}} & 0 \cr 0 & 0 & e^{- i a} 
 \end{pmatrix}
 \begin{pmatrix}\left( \cos(\omega r) {\bf 1}_2+\frac{\sin(\omega r)}{\omega} \Omega \right) & 0 \cr 0 & 1 \end{pmatrix}, 
 $$
 \end{widetext}
 
 where $\Omega:=\begin{pmatrix} - \frac{ia}{2r} & 1 \cr -1 & \frac{ia}{2r}\end{pmatrix}$ and $\omega:=\sqrt{1+\frac{a^2}{4r^2}}$. Since $ e^{T(-A+P)T^\dagger}$ must be block diagonal (with diagonal blocks of dimension $1$ and $2$ in that order) we must have $\sin(\omega r)=0$, which gives the first condition relating $a$ and $r$, that is,  
 \begin{equation}\label{condir}
 \omega r=\frac{\sqrt{4r^2+a^2}}{2}=l\pi. 
 \end{equation} 
 With this condition we distinguish two cases: $l$ odd and $l$ even. 
 
 \vspace{0.25cm}
 
 \noindent{\bf $l$ odd}
 
  \vspace{0.25cm}

 We have 
 $$
 Te^{(-A+P)} T^\dagger =\begin{pmatrix} -e^{ i\frac{a}{2}} & 0 & 0 \cr 0 & -e^{ i\frac{a}{2}} & 0 \cr 0 & 0 & e^{- ia} \end{pmatrix}, 
 $$
 so that 
 \begin{align*}
     e^{(-A+P)}&=T^\dagger \begin{pmatrix} -e^{i\frac{a}{2}} & 0 & 0 \cr 0 & -e^{ i\frac{a}{2}} & 0 \cr 0 & 0 & e^{-ia} \end{pmatrix} T\\
     &=\begin{pmatrix} 
 -e^{ i \frac{a}{2}} & 0 \cr 0 & - i e^{- \frac{ia}{4}} K \end{pmatrix}, \notag
 \end{align*}
 with $K:=\begin{pmatrix} \sin(\frac{3a}{4}) & \cos(\frac{3a}{4} )\cr - \cos(\frac{3a}{4}) & \sin (\frac{3a}{4}) \end{pmatrix}$. Calculating $e^Ae^{-A+P}$, we obtain
 $$
 e^A e^{-A+P}=\begin{pmatrix} 
 -e^{ i \frac{a}{2}} & 0 \cr 0 & - i e^{- \frac{ia}{4}} \tilde  K \end{pmatrix}, 
 $$
 with $\tilde K=\begin{pmatrix} \cos(a) & \sin(a) \cr -\sin(a) & \cos(a) \end{pmatrix}K=\begin{pmatrix} \cos(\frac{\pi}{2} +\frac{a}{4}) & \sin (\frac{\pi}{2} +\frac{a}{4}) \cr 
 -\sin( \frac{\pi}{2} +\frac{a}{4})&  \cos(\frac{\pi}{2} +\frac{a}{4}) \end{pmatrix}$. In general,  direct verification shows that with $\rho_0$ and $\rho_1$ given in (\ref{rho0}) (\ref{realrho}), the set of matrices $X$  in $SU(3)$ such that $X\rho_0X^\dagger=\rho_1$ is given by 
 \begin{widetext}
     {\small{
\begin{equation}\label{esserho}
\left\{  X\in SU(3) \, | \, X=\begin{pmatrix} e^{-i \hat \phi} & 0 & 0 \cr 0 & \cos(\theta)e^{i \lambda} & \sin(\theta) e^{i\mu} \cr 0 & -\sin(\theta)e^{i\lambda} & \cos(\theta) e^{i\mu} \end{pmatrix}, \,   \frac{1}{2} -\frac{(1-2a)}{2} \cos(2\theta)=b, \frac{(1-2a)}{2} \sin(2\theta) =N  \right\}. 
\end{equation}}}
 \end{widetext}

The condition of formula (\ref{esserho}) for our problem gives $\theta=\frac{\pi}{4}+k\pi$ for integer $k$ and using $\theta=\frac{\pi}{2}+\frac{a}{4}$ as in the previous calculation, we get the condition for $a$
 \begin{equation}\label{condia}
 a=(4k-1)\pi. 
 \end{equation}
 Plugging this into (\ref{condir}) we have to find integers $k$ and odd $l$ to minimize 
 $$
 4r^2=(4l^2-(4k-1)^2)\pi^2 >0 
 $$
 The minimum is achieved for $k=0$ and $l=1$ which gives $3\pi^2$ on the right hand side. If $k>0$, the condition $4l^2-(4k-1)^2>0$ gives $|l| >- \frac{1}{2}+2k$, and since $l$ is an (odd) integer, we have $|l|\geq 2k$, which gives $4l^2-(4k-1)^2 \geq 8k-1\geq 7>3$. Analogously one sees that if $k<0$ we obtain a value $> 3$. 

   \vspace{0.25cm}
 
 \noindent{\bf $l$ even}
 
   \vspace{0.25cm}
 
 The treatment is similar to the $l$ odd case.   
 We have 
 $$
 Te^{(-A+P)} T^\dagger =\begin{pmatrix} e^{ i\frac{a}{2}} & 0 & 0 \cr 0 & e^{ i\frac{a}{2}} & 0 \cr 0 & 0 & e^{- ia} \end{pmatrix}, 
 $$
 so that 
 $$
 e^{(-A+P)}=T^\dagger \begin{pmatrix} e^{ i\frac{a}{2}} & 0 & 0 \cr 0 & e^{ i\frac{a}{2}} & 0 \cr 0 & 0 & e^{- ia} \end{pmatrix} T=\begin{pmatrix} 
 e^{ i \frac{a}{2}} & 0 \cr 0 &  e^{- \frac{ia}{4}} H \end{pmatrix}, 
 $$
 with $H:=\begin{pmatrix} \cos(\frac{3a}{4}) & - \sin(\frac{3a}{4} )\cr  \sin(\frac{3a}{4}) & \cos (\frac{3a}{4}) \end{pmatrix}$. Calculating $e^Ae^{-A+P}$, we obtain
 $$
 e^A e^{-A+P}=\begin{pmatrix} 
 e^{ i \frac{a}{2}} & 0 \cr 0 &  e^{- \frac{ia}{4}} \tilde  H \end{pmatrix}, 
 $$
 with $\tilde H=\begin{pmatrix} \cos(a) & \sin(a) \cr -\sin(a) & \cos(a) \end{pmatrix} H=\begin{pmatrix} \cos(\frac{a}{4}) & \sin (\frac{a}{4}) \cr 
 -\sin( \frac{a}{4})&  \cos(\frac{a}{4}) \end{pmatrix}$.Using again the condition from  formula (\ref{esserho}) for our problem gives $\theta=\frac{\pi}{4}-k\pi$ for integer $k$ and using $\theta=\frac{a}{4}$ as in the previous calculation, we get the condition for $a$
 \begin{equation}\label{condiabis}
 a={\pi}+4k\pi. 
 \end{equation}
 Plugging this into (\ref{condir}) we have to find integers $k$ and {\it even} $l$ to minimize 
 $$
 4r^2=(4l^2-(1+4k)^2)\pi^2 >0. 
 $$
We notice that, if $k\geq 0$, for fixed value of $k$, the optimal admissible value of even $l$ is $2k+2$ which gives $(4l^2-(1+4k)^2)\geq 24k+15\geq 15$. Therefore this is not the optimal. If $k <0$, the optimal value for $|l|$ is $|l|=3|k|$ which gives $4l^2-(1+4k)^2=16k^2-4l^2-(1+4k)^2)\geq 7 >3$. Thus, this also not optimal.

 \subsection{Summary of Cases 1 and 2 and conclusion of the proof}

It follows from the previous discussion that the minimal cost $r^2$ is found in the Case $2$ of Subsection \ref{possnsc}  since we found a value of $J=r^2=\frac{3}{4}\pi^2$ smaller than the minimum value found in Case $1$ of Section \ref{scalar1}  which was of $J=r^2=\frac{15}{16}\pi^2$. Using the corresponding value $r=\frac{\sqrt{3}\pi}{2}$ and the value $a$ obtained from (\ref{condia}) with $k=0$ as described above, which gives the matrices of Theorem \ref{Final}.

\bibliography{PMPDensity.bib}

\end{document}